\documentclass{article}
\usepackage[utf8]{inputenc}
\usepackage{hyperref}
\usepackage{fullpage}
\usepackage{amssymb}
\usepackage{amsthm}
\usepackage{float}
\usepackage{esvect}
\usepackage{stmaryrd}
\usepackage{subcaption}
\usepackage{mathtools}
\usepackage[dvipsnames]{xcolor}
\usepackage{tikz}
\usetikzlibrary{calc}
\usetikzlibrary{shapes.geometric}
\usetikzlibrary{shapes.misc}
\usetikzlibrary{positioning}
\newcommand{\bvdots}{ \tikz[baseline, every node/.style={inner sep=0}]{ \node at (0,0){.}; \node at (0,-6pt){.}; \node at (0,6pt){.}; } }
\newcommand{\quot}[2]{{\raisebox{.2em}{$#1$}\left/\raisebox{-.2em}{$#2$}\right.}}
\newcommand{\interp}[1]{\left\llbracket#1\right\rrbracket}

\newcommand{\N}{\mathbb{N}_0}
\newcommand{\df}{\stackrel{def}{=}}

\usepackage{tikzit}

\definecolor{zx_grey}{RGB}{211,211,211}
\tikzstyle{gn}=[fill=green, draw=black, shape=circle, tikzit category=ZX, tikzit fill=green, tikzit draw=black, tikzit shape=circle, inner sep=0.1em]
\tikzstyle{rn}=[fill=red, draw=black, shape=circle, tikzit fill=red, tikzit draw=black, tikzit category=ZX, tikzit shape=circle, inner sep=0.1em]
\tikzstyle{divide}=[regular polygon, regular polygon sides=3, shape border rotate=90, draw=black, fill={zx_grey}, inner sep=1.5pt, tikzit category=scal, rounded corners=0.8mm]
\tikzstyle{black}=[fill=black, draw=black, shape=circle, tikzit fill=black, tikzit draw=black, tikzit shape=circle, tikzit category=IH, inner sep=2pt]
\tikzstyle{gather}=[fill={zx_grey}, draw=black, tikzit category=scal, rounded corners=0.8mm, regular polygon, regular polygon sides=3, shape border rotate=-90, inner sep=1.5pt]
\tikzstyle{ggen}=[fill=white, draw=black, shape=rectangle, rounded corners=2mm, line width=1pt, tikzit draw=red, tikzit category=scal]
\tikzstyle{white}=[fill=white, draw=black, shape=circle, inner sep=2pt, tikzit category=IH]
\tikzstyle{mbox}=[fill=white, draw=black, rounded rectangle, rounded rectangle west arc=none, tikzit category=scal, tikzit shape=rectangle]
\tikzstyle{A}=[fill=white, shape=circle, tikzit category=scal, inner sep=1pt]
\tikzstyle{ggreen}=[fill=green, draw=black, shape=circle, tikzit category=SZX, tikzit fill=green, tikzit draw=black, line width=1pt, inner sep=0.1em]
\tikzstyle{gred}=[fill=red, draw=black, shape=circle, rounded corners=2mm, tikzit category=SZX, inner sep=0.1em, tikzit fill=red, line width=1pt]
\tikzstyle{ghad}=[fill=yellow, draw=black, shape=rectangle, tikzit category=SZX, tikzit shape=rectangle, tikzit fill=yellow, inner sep=0.1em, line width=1pt]
\tikzstyle{boxm}=[fill=white, draw=black, rounded rectangle, tikzit category=scal, tikzit shape=rectangle, rounded rectangle east arc=none]
\tikzstyle{box}=[fill=white, draw=black, shape=rectangle]
\tikzstyle{had}=[fill=yellow, draw=black, shape=rectangle, tikzit category=ZX, tikzit fill=yellow, tikzit draw=black, inner sep=0.1em]
\tikzstyle{gwhite}=[fill=white, draw=black, shape=circle, tikzit fill=white, tikzit shape=circle, line width=1 pt, inner sep=2 pt, tikzit draw=red]
\tikzstyle{gblack}=[fill=black, draw=black, shape=circle, tikzit fill=black, tikzit shape=circle, line width=1 pt, inner sep=2 pt, tikzit draw=red]
\tikzstyle{antipode}=[fill=red, draw=black, shape=rectangle, tikzit fill=red, tikzit draw=black, tikzit shape=rectangle, inner sep=2pt]
\tikzstyle{diamond}=[fill=white, draw=black, shape=diamond, inner sep=2pt]
\tikzstyle{mongr}=[fill=green, draw=green, shape=circle, inner sep=2pt]
\tikzstyle{monbl}=[fill=blue, draw=blue, shape=circle, inner sep=2pt]
\tikzstyle{bg}=[inner sep=0.7mm, minimum width=0pt, minimum height=0pt, fill=green, draw=white, very thick, shape=circle]
\tikzstyle{br}=[inner sep=0.7mm, minimum width=0pt, minimum height=0pt, fill=red, draw=white, very thick, shape=circle]
\tikzstyle{rmat}=[draw, signal, fill={zx_grey}, signal to=east, signal from=west, inner sep=1pt, minimum height=6pt]
\tikzstyle{lmat}=[draw, signal, fill={zx_grey}, signal to=west, signal from=east, inner sep=1pt, minimum height=6pt]
\tikzstyle{umat}=[draw, signal, fill={zx_grey}, signal to=north, signal from=south, inner sep=1pt, minimum width=6pt]
\tikzstyle{dmat}=[draw, signal, fill={zx_grey}, signal to=south, signal from=north, inner sep=1pt, minimum width=6pt]
\tikzstyle{box}=[shape=rectangle, text height=1.5ex, text depth=0.25ex, yshift=0.5mm, fill=white, draw=black, minimum height=5mm, yshift=-0.5mm, minimum width=5mm, font={\small}]
\tikzstyle{Z dot}=[inner sep=0mm, minimum size=2mm, shape=circle, draw=black, fill={rgb,255: red,160; green,255; blue,160}]
\tikzstyle{gdot}=[minimum size=3mm, font={\scriptsize\boldmath}, shape=rectangle, rounded corners=1.3mm, inner sep=1mm, outer sep=-1.8mm, scale=0.8, tikzit shape=circle, draw=black, fill=green, tikzit draw=blue]
\tikzstyle{X dot}=[Z dot, shape=circle, draw=black, fill={rgb,255: red,220; green,0; blue,0}]
\tikzstyle{rdot}=[minimum size=3mm, font={\scriptsize\boldmath}, shape=rectangle, rounded corners=1.3mm, inner sep=1mm, outer sep=-1.8mm, scale=0.8, tikzit shape=circle, draw=black, fill=red, tikzit draw=blue]
\tikzstyle{grdot}=[minimum size=3mm, font={\scriptsize\boldmath}, shape=rectangle, rounded corners=1.3mm, inner sep=1mm, line width=1pt, outer sep=-1.5mm, scale=0.8, tikzit shape=circle, draw=black, fill=red, tikzit draw=blue]
\tikzstyle{ggdot}=[minimum size=3mm, font={\scriptsize\boldmath}, shape=rectangle, line width=1pt, rounded corners=1.3mm, inner sep=1mm, outer sep=-1.5mm, scale=0.8, tikzit shape=circle, draw=black, fill=green, tikzit draw=blue]
\tikzstyle{arrow}=[-->]
\tikzstyle{new style 0}=[fill=yellow, draw=black, shape=triangle]

\tikzstyle{arrow}=[->]
\tikzstyle{very thick}=[-, line width=1pt, tikzit draw=red]
\tikzstyle{pointille}=[dashed, -]
\tikzstyle{red}=[-, draw=red]
\tikzstyle{blue}=[-, draw=blue]
\tikzstyle{green}=[-, draw=green]
\tikzstyle{strike}=[-, tikzit draw={rgb,255: red,191; green,0; blue,64}, strike through]
\tikzstyle{strike'}=[-, tikzit draw=cyan, strike bend]
\tikzstyle{dashed arrow}=[->, tikzit draw=green, draw=black, dashed]

\title{Colored props for large scale graphical reasoning}

\author{Titouan Carette\\
	CNRS, LORIA, Inria Mocqua, Universit\'e de Lorraine, F 54000 Nancy, France\\
	titouan.carette@loria.fr \\[0.5cm]
	Simon Perdrix\\
	CNRS, LORIA, Inria Mocqua, Universit\'e de Lorraine, F 54000 Nancy, France\\
	simon.perdrix@loria.fr}

\newtheorem{definition}{definition}
\newtheorem{theorem}{theorem}
\newtheorem{lemma}{lemma}
\begin{document}

\maketitle

\begin{abstract}
The prop formalism allows representation of processes with string diagrams and has been successfully applied in various areas such as quantum computing, electric circuits and control flow graphs. However, these graphical approaches suffer from scalability problems when it comes to writing large diagrams. A proposal to tackle this issue has been investigated for ZX-calculus using colored props. This paper extends the approach to any prop, making it a general tool for graphical languages manipulation.
\end{abstract}

%%Eviter les citations dans les abstracts..

%% enoncer plus pr�cisement les r�sultats ? 

%\textbf{motivation}
\vspace{0.5cm}
There is a long list of graphical representations of various processes: Petri nets \cite{bonchi2019diagrammatic}, control flow graphs \cite{bonchi2015full}, boolean circuits \cite{lafont2003towards}, quantum circuits \cite{nielsen2002quantum}, proof nets \cite{girard2017proof}, \emph{etc}. Most of them are unified into the prop formalism. A \textbf{prop} is a monoidal category with a free monoid of objects. Props were first introduced by Mac Lane in~\cite{maclane1965categorical}. They are now used as a foundation for graphical reasoning through string diagrams. Each diagram is seen as a morphism in a prop. Props admit presentations by generators and equations~\cite{baez2018props}.

Compared to usual symbolic manipulations, diagrams have the advantage of capturing some fondamental properties in a simple and intuitive way. In practice, most axioms are encoded into the topology making the notation lighter and easier to read. In the specific case of quantum computation,  another advantage is the sometime radical compression allowing to represent matrices scaling exponentially with the number of qubits, by diagrams with a polynomial number of nodes. However, very large diagrams are difficult to draw, read and manipulate. Uniform diagrammatical proofs often use a lot of dots to suggest a repeating pattern or arbitrary large structures. The recent applications of these graphical methods and the need for verified proofs which cannot tolerate ellipsis lead to the development of new notations, allowing graphical languages to scale up.

\textbf{Contributions.}
For every graphical language we provide a way to construct a \emph{colored} graphical language where each wire is indexed by a size (its \emph{color}), and can be interpreted as several wires side by side. 
We call this the \emph{scalable construction}. We show that the scalable construction enjoys a universal property which gives us a natural way to extend the semantics of any (monochromatic) prop into semantics for the scalable version.
We then describe how other props corresponding to substructures can be embedded into a scalable prop by introducing a \emph{box construction}. This can be viewed as an abstraction of specific large graphical structures.
Finally, we use previously known completeness results for graphical structures like monoids, bialgebras and Interacting Hopf Algebras, to provide a compressed way to deal with those structures in other graphical languages.

\textbf{Related works.}
Our approach has been initiated in the special case of the ZX-calculus in \cite{carette2019szx}, we extend this construction to any graphical language.
The scalability problems for graphical languages has also motivated the introduction of !-boxes in \cite{kissinger2016tensors}. Unlike our work, !-boxes do not provide a new graphical language but a way to represent families of diagrams and to reason uniformly with them. In practice, scalable props provide a compact representation of large diagrams while !-boxes give a graphical description of how to construct a family of diagrams. Divider and gatherer generators  similar to ours have been used in \cite{miatto2019graphical} to graphically define  various matrix products. We share with this work our definition of direct sum of matrix arrows. The definition of our scaled generators matches what has been called monoidal multiplexing in \cite{chantawibul2018monoidal}, another work on the scalability problem. In fact this is the natural way to define the action of a diagram on many wires.

\textbf{Structure of the paper.} 
The first section sets up the prop framework to describe colored graphical languages. It provides the necessary definitions and results needed for the different scalable constructions.
In the second section, %the set of colors is set to $\N$, the non-zero positive integers. T
the scalable constructions are described together with their main properties.
The last section provides examples of large scale graphical structures that can benefit from our construction. We also give examples of concrete models for those structures, mainly taken from the ZX-, ZH- and ZW-calculi, introduced respectively in \cite{coecke2011interacting}, \cite{backens2018zh} and \cite{hadzihasanovic2018zw}.

\section{Colored graphical languages}

This section lays out the mathematical foundation for colored graphical reasoning. Nothing here is new, we just restate for graphical languages well known results of categorical universal alegebra.

\subsection{Colored props}

We work in the setting of colored props.

\begin{definition}[Colored prop]
	A \textbf{colored prop} is a small symmetric strict monoidal category $\mathbf{P}$ together with a set of \textbf{colors} $C$, such that the set of objects of $\mathbf{P}$ is freely spanned by the elements of $C$.
\end{definition}

We use the term \textbf{morphism} to denote an arrow of a colored prop as a category. From now on, we will use the term prop for a colored prop. A prop is \textbf{$C$-colored} when the set of colors is $C$ and \textbf{monochromatic} when $C$ is a singleton.
The set of objects of a $C$-colored prop is $C^*$, the set of finite lists of elements of $C$. A list of $n$ colors is denoted $\langle c_i\rangle_{i=1}^{n}$.  We will write $0$ for the empty list and $\boxtimes$ for the concatenation. This coincides with the monoidal structure in a $C$-colored prop. We have $0\boxtimes \mathbf{a} = \mathbf{a}\boxtimes 0 = \mathbf{a}$ for all $\mathbf{a}\in C^*$. Given a color map $\phi:C\to C'$ we denote $\phi^*:C^*\to {C'}^*$ its extension to lists defined by $\phi^*(0)=0$, $\phi^*(\mathbf{a}\boxtimes \mathbf{b})=\phi^*(\mathbf{a})\boxtimes \phi^*(\mathbf{b})$ and $\phi^*(\langle c\rangle)=\langle\phi(c) \rangle$.

A \textbf{prop morphism} is a symmetric strict monoidal functor mapping colors to colors. Formally, a prop morphism is a color map $\phi:C\to C'$ and a symmetric strict monoidal functor $F:\textbf{P}\to \textbf{P'}$ such that $F(\mathbf{a})=\phi^*(\mathbf{a})$. \textbf{Prop}  is the category of props and prop morphisms. This category is complete and co-complete. See \cite{hackney2015category} for a study of the category of props.

We will mainly work in subcategories of $\textbf{Prop}$ where the set of colors is fixed. A \textbf{$C$-colored prop morphism} is a prop morphism between two $C$-colored props where $\phi=id_C$. For any set $C$ of colors,  \textbf{$C$-Prop} is the category of $C$-colored props and $C$-colored prop morphisms. \textbf{$C$-Prop} is a subcategory of \textbf{Prop} which is not full. By setting $C$ to be a singleton we recover the category of monochromatic props and monochromatic prop morphisms that is often simply called \textbf{Prop} in the literature.

We will use the term \textbf{functor} when using functors in $\textbf{SymCat}$ (the category of small symmetric monoidal category and symmetric monoidal functors), which in general are not prop morphisms.

We now mention some remarkable props. $\mathbf{1}_C$ is the codiscrete category over $C^*$, it is a terminal object in \textbf{$C$-Prop}. The terminal object for monochromatic props is denoted $\mathbf{N}$. $\mathbb{P}_C$ is the category of permutations over finite lists of colors in $C$. It is an initial object in \textbf{$C$-Prop}. The initial object for monochromatic props is denoted $\mathbb{P}$.

\subsection{String diagrams}

Colored props admit a nice graphical representation with colored string diagrams. Let $\textbf{P}$ be a $C$-colored prop. %, the provided examples will use $C=\{\textcolor{red}{\bullet},\textcolor{blue}{\bullet}\}$. 
The idea is that each color $c\in C$ corresponds to a color of wire. In this sub-section, we will actually use various colors  ($C=\{\textcolor{red}{\bullet},\textcolor{blue}{\bullet},\textcolor{green}{\bullet}\}$) to represent the wires:  the identity map $c\to c$ is represented as a wire of color $c$, e.g.\begin{tikzpicture}
	\begin{pgfonlayer}{nodelayer}
		\node [style=none] (0) at (0, 0) {};
		\node [style=none] (1) at (1, 0) {};
		\node [style=none] (2) at (1, 0) {};
	\end{pgfonlayer}
	\begin{pgfonlayer}{edgelayer}
		\draw [style=red] (0.center) to (2.center);
	\end{pgfonlayer}
\end{tikzpicture}
. Each morphism $\langle c_i\rangle_{i=1}^{n}\to \langle c'_j\rangle_{j=1}^{m}$ is represented as a diagram with $n$ input wires colored according to the $c_i$s and $m$ output wires colored according to the $c'_j$s. Here is an example for a map: $f:\langle\textcolor{red}{\bullet},\textcolor{blue}{\bullet}\rangle \to \langle\textcolor{blue}{\bullet},\textcolor{red}{\bullet}\rangle$: \input{fbox.tikz}.
The empty object $0$ and its identity map $id_0:0\to0$ correspond to the empty diagram\!\input{emptydiag.tikz}\!\!. Diagrams of type $0\to 0$ have no inputs nor outputs, we call them \textbf{scalars}.

In this representation the tensor product of two morphisms is represented by juxtaposition: $\input{fbox.tikz}\boxtimes\input{gbox.tikz}=\input{ftgbox.tikz}$. Thus the identity of the list $\langle\textcolor{red}{\bullet},\textcolor{blue}{\bullet}\rangle$, $id_{\langle\textcolor{red}{\bullet},\textcolor{blue}{\bullet}\rangle}=id_{\langle\textcolor{red}{\bullet}\rangle}\boxtimes id_{\langle\textcolor{blue}{\bullet}\rangle}$, is represented as: \input{idredblue.tikz}.

The composition of two morphisms is done by plugging the corresponding colored wires: $\input{fbox.tikz}\circ\input{gbox.tikz}=\input{fogbox.tikz}$.

The symmetry maps are represented by crossing the corresponding colored wires, for example $\sigma_{\langle\textcolor{red}{\bullet},\textcolor{blue}{\bullet}\rangle}: \langle\textcolor{red}{\bullet},\textcolor{blue}{\bullet}\rangle \to \langle\textcolor{blue}{\bullet},\textcolor{red}{\bullet}\rangle$ is represented as : \input{swapredblue.tikz}. Its inverse is $\sigma_{\langle\textcolor{blue}{\bullet},\textcolor{red}{\bullet}\rangle}: \langle\textcolor{blue}{\bullet},\textcolor{red}{\bullet}\rangle \to \langle\textcolor{red}{\bullet},\textcolor{blue}{\bullet}\rangle$, graphically: $\input{swinv0.tikz}=\input{swinv1.tikz}$.

In this representation, the axioms of symmetric monoidal categories correspond to topological properties. The naturality of the symmetry corresponds to equations of the form: $\input{nats0.tikz}=\input{nats1.tikz}$. If two diagrams with colored wires are isomorphic then the corresponding morphisms are equal according to the axioms of $C$-colored props. See \cite{selinger2010survey} for an overview of the different ways to draw categories.

\subsection{Graphical languages}

We describe props with their presentations by generators and equations. We call such an equational theory a graphical language. We start by defining colored signatures.

\begin{definition}[Signature]
	A colored \textbf{signature} $\Sigma$ is a set of colors $C$ together with a family of sets $\Sigma[\mathbf{a},\mathbf{b}]$ indexed by $\mathbf{a},\mathbf{b}\in C^*$. The elements of $\Sigma[\mathbf{a},\mathbf{b}]$ are called \textbf{generators} of type $\mathbf{a}\to \mathbf{b}$. We write $|\Sigma|$ for $\biguplus\limits_{\mathbf{a},\mathbf{b}} \Sigma[\mathbf{a},\mathbf{b}]$, the set of all generators.
\end{definition}

A signature is a way to present a collection of morphisms of a colored prop that will be used later as building blocks, hence the name generators.
We use the same terminology as props, in particular signature stands for colored signature and we say that a signature is monochromatic when the set of colors is a singleton. A way to present a signature is as a functor $C^*\times C^* \to \textbf{Set}$, where $C^*\times C^*$ is a discrete category. Thus, there is a category of $C$-colored signatures $C$-$\mathbf{Sig}$ which is just another name for the functor category $\textbf{Set}^{C^* \times C^*}$. The \textbf{$C$-colored signature morphisms} are natural transformations between signatures. In other words $\alpha:\Sigma\to \Sigma'$ is a familly of functions $\alpha_{\mathbf{a},\mathbf{b}}:\Sigma[\mathbf{a},\mathbf{b}]\to \Sigma'[\mathbf{a},\mathbf{b}]$.
%$C$-$\mathbf{Sig}$ is complete and cocomplete (taking limits and co-limits pointwise), and furthermore assuming the axiom of choice every epimorphism split (roughly speaking, all surjective functions has a left inverse).

The interest of this definition comes from the following result presented in \cite{baez2018props}.

\begin{theorem}\cite{baez2018props}
	Let $U_C:C\text{-}\textbf{Prop}\to C\text{-}\mathbf{Sig}$ be the forgetful functor sending a $C$-colored prop $\textbf{P}$ to the $C$-colored signature $U_C(\textbf{P})$ such that for each $a,b\in C^*$ $\Sigma[a,b]=\textbf{P}[a,b]$. $U$ has a left adjoint $F:C\text{-}\mathbf{Sig}\to C\text{-}\textbf{Prop}$ and $C\text{-}\textbf{Prop}$ is equivalent to the Eilenberg-Moore category of the monad $T\df U\circ F$. We denote $\eta: 1 \Rightarrow T$ the unit, $\epsilon: F\circ U \Rightarrow 1$ the counit  and $\mu: T^2 \Rightarrow T$ the multiplication.
\end{theorem}

The main interest of the theorem is to provide us with the free functor $F$ sending a $C$-colored signature $\Sigma$ to the free $C$-colored prop $F(\Sigma)$ spanned by this signature.
Intuitively, a morphism in the free prop $F(\Sigma)$ is a string diagram built from generators in $\Sigma$ linked by wires and swaps. From now on, by \textbf{diagram} we mean a morphism in $F(\Sigma)$. Such diagrams are identified up to the topological moves corresponding to the symmetric monoidal axioms. $F$ can be used to define equational theories for props.

\begin{definition}[equation]
	An \textbf{equation} of type $a\to b$ with respect to a signature $\Sigma$ is a pair $(f,g)$ where $f,g\in F(\Sigma)[a,b]$.
\end{definition}

A set of equations over a signature $\Sigma$ can be presented as a pair of signature morphisms $l,r: E\to UF(\Sigma)$. A graphical calculus is then formed by a signature and a set of equations over this signature. 

\begin{definition}[$C$-colored graphical language]
	A $C$-colored graphical language is a tuple $\mathcal{L}\df(\Sigma_\mathcal{L}, E_\mathcal{L}, l_\mathcal{L}, r_\mathcal{L})$ where $\Sigma_\mathcal{L}$ and $E_\mathcal{L}$ are $C$-colored signatures and $l_\mathcal{L}$ and $r_\mathcal{L}$ are $C$-colored signature morphisms of type $E_\mathcal{L}\to UF(\Sigma_\mathcal{L})$.
\end{definition}

To any graphical language $\mathcal{L}\df(\Sigma_\mathcal{L}, E_\mathcal{L}, l_\mathcal{L}, r_\mathcal{L})$ corresponds an \textbf{underlying} prop $\textbf{L}$, defined as the coequalizer: \input{coeqp.tikz}.

It is defined only up to prop isomorphism. Given $f,g\in F(\Sigma_\mathcal{L})[n,m]$ we write $\mathcal{L}\vdash f=g$ iff $\pi_\mathcal{L}(f)=\pi_\mathcal{L}(g)$ in $\textbf{L}$, this means that using the equations as local rewriting rules we can transform the diagram $f$ into $g$.

We say that two graphical languages $\mathcal{L}$ and $\mathcal{Y}$ are \textbf{equipotent} iff $\textbf{L}\simeq \textbf{Y}$.

We define a \textbf{translation} between two graphical calculi as a signature morphism $\theta: \Sigma_\mathcal{L}\to UF(\Sigma_\mathcal{Y})$ satisfying the \textbf{soundness condition}: $\pi_\mathcal{Y}\circ \left(\epsilon_{F\Sigma_\mathcal{Y}}\circ F\theta\right)\circ \left(\epsilon_{F\Sigma_\mathcal{L}}\circ Fr_\mathcal{L}\right)=\pi_\mathcal{Y}\circ \left(\epsilon_{F\Sigma_\mathcal{Y}}\circ F\theta\right)\circ \left(\epsilon_{F\Sigma_\mathcal{L}}\circ Fl_\mathcal{L}\right)$. This asserts that equivalent diagrams in $\mathcal{L}$ are sent to equivalent diagrams in $\mathcal{Y}$. The coequalizer property of $\textbf{Y}$ gives a unique prop morphism $Bold(\theta
)$ such that $Bold(\theta)\circ \pi_\mathcal{L}=\pi_\mathcal{Y}\circ \left(\epsilon_{F\Sigma_\mathcal{Y}}\circ F\theta\right)$. In fact, the soundness condition is equivalent to the existence of $Bold(\theta)$.

\begin{lemma}\label{glcat}
	There is a category $\textbf{GL}$ of graphical languages and translations and an essentially surjective full functor $Bold:\textbf{GL}\to\textbf{Prop}$.
\end{lemma}

\noindent A proof is given in \ref{glcat:pr}. 

\vspace{0.2cm}

This functor allows us to define props and prop morphism using graphical languages and translations. Furthermore any prop and prop morphism admit such a description, but it is not unique.

$C$-\textbf{Prop} is a cocomplete category thus it has sums and coequalizers. These can be described with graphical languages:

\begin{definition}[Sum of graphical languages]
	The sum of two $C$-colored graphical languages $\mathcal{L}$ and $\mathcal{Y}$ is defined as $\mathcal{L}+\mathcal{Y}\df \left(\Sigma_\mathcal{L}+\Sigma_{\mathcal{Y}}, E_\mathcal{L}+E_{\mathcal{Y}},l_{\mathcal{L}+\mathcal{Y}}  ,r_{\mathcal{L}+\mathcal{Y}}\right)$, where $l_{\mathcal{L}+\mathcal{Y}}=[UF(\iota_{\Sigma_\mathcal{L}})\circ l_\mathcal{L},UF(\iota_{\Sigma_\mathcal{L}})\circ l_\mathcal{Y}]$ and $r_{\mathcal{L}+\mathcal{Y}}=[UF(\iota_{\Sigma_\mathcal{L}})\circ r_\mathcal{L},UF(\iota_{\Sigma_\mathcal{L}})\circ r_\mathcal{Y}]$.
\end{definition}

The sum of two graphical languages has the generators and equations of both languages.

\begin{lemma}\label{copr}
	$Bold(\mathcal{L}+\mathcal{Y})$ is a coproduct $\textbf{L}+\textbf{Y}$ in $C$-\textbf{Prop}.
\end{lemma}

\noindent A proof is given in \ref{copr:pr}. 

\vspace{0.2cm}

Usually when we build new graphical languages we take the sum of other graphical languages and then add more equations.

\begin{definition}[Quotient of a graphical language by equations]
	Given a $C$-colored graphical languages $\mathcal{L}$ and a set of equations $\left(E,l,r\right)$ over $\Sigma_\mathcal{L}$, the quotient of $\mathcal{L}$ by $E$ is defined as $\quot{\mathcal{L}}{E}\df \left(\Sigma_\mathcal{L},E_\mathcal{L}+E,[l_\mathcal{L},l],[r_\mathcal{L},r] \right)$.
\end{definition}

\begin{lemma}\label{coeq}
	$Bold(\quot{\mathcal{L}}{E})$ is a coequalizer of $\pi_\mathcal{L}\circ \epsilon_{F\Sigma_\mathcal{L}}  \circ Fl$ and $\pi_\mathcal{L}\circ \epsilon_{F\Sigma_\mathcal{L}}  \circ Fr$ in $C$-\textbf{Prop}.
\end{lemma}

\noindent A proof is given in \ref{coeq:pr}.

\vspace{0.2cm}

The semantics of a graphical language is given by interpretation functors.

\begin{definition}
	An \textbf{interpretation} of a graphical language $\mathcal{L}$ is a functor $\interp{\_}:\textbf{L}\to \textbf{C}$. A language is said \textbf{complete}, respectively \textbf{universal}, for the category $\textbf{C}$ if there exists an interpretation which is full, respectively faithful.
\end{definition}

A natural question, given a graphical language, is to find an interpretaion for which the language is universal and complete. Such examples will be given in section 3. We are now ready to introduce the scalable construction.

\section{The scalable construction}

From now on, our set of colors will always be $\N:=\mathbb N \setminus \{0\}$. In this setting we will say \textbf{size} instead of color. We will use string diagrams to represent morphisms. A wire of size $1$ is said \textbf{simple} and a wire of size $n>1$ is said \textbf{big}. To simplify the notations, a thin wire will always be simple when a bold one without label can be of any size. We will only use size labels when absolutely necessary.

The objects of an $\N$-colored prop are lists of positive integers. The  list with $k$ times the integer $n$: $\langle n\rangle_{i=1}^k$ is denoted $k\cdot n$. We have $0\cdot n = k \cdot 0 = 0$ but $k\cdot n \neq n \cdot k$ and in particular $k\cdot 1 \neq 1 \cdot k$. In fact $k\cdot 1$ corresponds to $k$ simple wires and $1\cdot k$ to one big wire of size $k$. Given a family of wires of arbitrary size, which corresponds to an arbitrary object in an $\N$-colored prop, we define a notion of global size.

\begin{definition}
	Given an $\N$-colored prop $\textbf{P}$, the \textbf{global size} functor $[\_]:\textbf{P}\to \mathbf{N}$ is the unique functor satisfying: $[0]=0$ and $[1\cdot n]=n$.
\end{definition}

Intuitively big wires of size $k$ are ribbon cables, representing $k$ simple wires together, the global size functor just counts the overall number of simple wires. This intuition is made precise by the wire calculus.

\subsection{The wire calculus}

\begin{definition}[$\mathcal{D}$ and $\mathcal{G}$]
	The graphical languages $\mathcal{D}$ and $\mathcal{G}$ (for dividers and gatherers) are respectively freely generated by the signatures $\Sigma_\mathcal{D}[1\cdot (n+1),(1\cdot 1) \boxtimes (1\cdot n)]\stackrel{def}{=}\{\delta_n\}$ and $\Sigma_\mathcal{G}[(1\cdot 1) \boxtimes (1\cdot n),1\cdot (n+1)]\stackrel{def}{=}\{\gamma_n\}$ for all $n\in\N$. The generator $\delta_n$ is called the \textbf{divider} of size $n$ and is depicted:  \input{deltan.tikz}. The generator $\gamma_n$ is called the \textbf{gatherer} of size $n$ and is depicted:  \input{gamman.tikz}. We take the conventions: $\delta_0 \stackrel{def}{=} id_1$ and $\gamma_0 \stackrel{def}{=} id_1$.
\end{definition}

We have $\textbf{G}\simeq\textbf{D}^{op}$. We then define how those props interact in the spirit of \cite{lack2004composing}. The \textbf{expansion equation} of size $n$, $exp_n$ is the equation $\gamma_n \circ \delta_n = id_{1\cdot (n+1)}$, pictorially: $\input{exp0.tikz}\quad\stackrel{exp\label{exp}}{=}\quad\begin{tikzpicture}
	\begin{pgfonlayer}{nodelayer}
		\node [style=none] (6) at (2, 0) {};
		\node [style=none] (9) at (0, 0) {};
	\end{pgfonlayer}
	\begin{pgfonlayer}{edgelayer}
		\draw [style=very thick] (9.center) to (6.center);
	\end{pgfonlayer}
\end{tikzpicture}
$. The set of all expansion equations for each $n\in \N$ is denoted $Exp$.
The \textbf{elimination equation} of size $n$, $elim_n$ is the equation $\delta_{n} \circ \gamma_{n} = id_{1\cdot 1} \boxtimes id_{1\cdot n}$, pictorially: $\input{elim0.tikz}\quad\stackrel{exp\label{elim}}{=}\quad\begin{tikzpicture}
	\begin{pgfonlayer}{nodelayer}
		\node [style=none] (2) at (0, -0.5) {};
		\node [style=none] (3) at (0, 0.5) {};
		\node [style=none] (7) at (2, 0.5) {};
		\node [style=none] (8) at (2, -0.5) {};
	\end{pgfonlayer}
	\begin{pgfonlayer}{edgelayer}
		\draw [style=very thick] (2.center) to (8.center);
		\draw (7.center) to (3.center);
	\end{pgfonlayer}
\end{tikzpicture}
$. The set of all elimination equations for each $n\in \N$ is denoted $Elim$.

The convention for $\delta_0$ and $\gamma_0$ makes $exp_0$ and $elim_0$ trivially true.

\begin{definition}[The wire calculus $\mathcal{W}$]
  The $\N$-colored graphical language $\mathcal{W}$ is defined as $\mathcal{W}\df\quot{\left(\mathcal{D}+\mathcal{G}\right)}{\left(Elim, Exp\right)}$.
\end{definition}

In $\mathcal{W}$, the role of dividers and gatherers is perfectly symmetric, thus we have $\textbf{W}\simeq\textbf{W}^{op}$. $\textbf{W}$ is a groupoid, in fact the \hyperref[elim]{elimination} and  \hyperref[exp]{expansion} rules exactly state that the generators are invertible. We also note that the generators preserve the global size so there is no morphism of type $\mathbf{a}\to \mathbf{b}$ when $[\mathbf{a}]\neq [\mathbf{b}]$. In fact we can go further: when $[\mathbf{a}]= [\mathbf{b}]$, the morphisms of type $\mathbf{a}\to \mathbf{b}$ are the permutations of $[\mathbf{a}]$. 

\begin{theorem}[Rewiring theorem]\label{thm:rw}
	$\textbf{W}$ is a full subcategory of the permutation category $\mathbb{P}_{\N}$, satisfying $[\mathbf{a}]\neq [\mathbf{b}]\Rightarrow \textbf{W}[\mathbf{a},\mathbf{b}]=\emptyset$.
\end{theorem}

\noindent A proof is given in \ref{pr:rw}.

\vspace{0.2cm}

This gives us a clear understanding of what $\textbf{W}$ looks like as a category. The rewiring theorem works as a coherence result in the sense that any equation is true up to permutations as soon as the types match. This gives us total freedom to rewire the way we want.

The way the dividers and gatherers are defined, taking the wires one by one, is useful to come up with a normal form but is still quite restrictive. The rewiring theorem allows us to unambiguously generalize dividers and gatherers to any wire size. Furthermore we can define a divider with an arbitrary number of outputs since the associativity equation holds: $\input{asdiv0.tikz}~=~\input{asdiv1.tikz}$. We define inductively $\Delta_n:1\cdot n\to n\cdot 1$ by $\Delta_{0}=id_0$, and $\Delta_{n+1}=\left(id_1 \otimes \Delta_{n}\right)\circ \delta_{n}$. The diagram $\Delta_n$ is a sequence of $n$ dividers of decreasing size. For any object $\mathbf{a}=\langle n_i\rangle_{i=1}^k$ we define $\Delta_\mathbf{a}\df\bigotimes_{i} \Delta_{n_i}$. 
Dually, we define $\Gamma_{m}:m\cdot 1\to 1\cdot m$ by $\Gamma_{0}=id_0$, and $\Gamma_{m+1}=\gamma_{m}\circ \left(id_1 \otimes \Gamma_{m}\right)$, which is a sequence of $m$ gatherers of increasing size. For any object $\mathbf{b}=\langle m_i\rangle_{i=1}^k$ we define $\Gamma_\mathbf{b}\df\bigotimes_{i} \Gamma_{m_i}$.

We now proceed to making a monochromatic prop interacting with dividers and gatherers through the scalable construction.

\subsection{The scalable construction}

Intuitively, the scalable construction is the free embedding of a monochromatic graphical language into the simple wires of $\textbf{W}$. 

\begin{definition}[$\mathcal{SL}$]
	Given a monochromatic graphical language $\mathcal{L}$, we define an $\N$-colored graphical language $\mathcal{L}^{\N}$ by:

	\noindent $\Sigma_{\mathcal{L}^{\N}}[a,b]\df \begin{cases} \Sigma_{\mathcal{L}}[n,m] & \!\!\!\text{if $(a,b){=}(n{\cdot}1, m{\cdot}1)$}\\\emptyset & \!\!\!\text{otherwise}\end{cases} \quad E_{\mathcal{L}^{\N}}[a,b]\df \begin{cases} E_{\mathcal{L}}[n,m] & \!\!\!\text{if $(a,b){=}(n{\cdot}1, m{\cdot}1)$}\\\emptyset & \!\!\!\text{otherwise}\end{cases}$
	
%	\begin{align*}
%		\Sigma_{\mathcal{L}^{\N}}[a,b]\df&\text{ if }  (a,b)=(n\cdot 1, m\cdot 1) \text{ then } \Sigma_{\mathcal{L}}[n,m]
%		\text{ else } \emptyset\\
%		E_{\mathcal{L}^{\N}}[a,b]\df&\text{ if }  (a,b)=(n\cdot 1, m\cdot 1) \text{ then } E_{\mathcal{L}}[n,m]
%		\text{ else } \emptyset
%	\end{align*}
	
	The \textbf{scalable} graphical calculus $\mathcal{SL}$ is defined as: $\mathcal{SL}\df \mathcal{L}^{\N}+\mathcal{W}$. The underlying prop $Bold(\mathcal{SL})$ is denoted $\mathcal{S}\textbf{L}$.
\end{definition}

The scalable construction was first introduced in \cite{chancellor2016graphical} to allow a compact representation of large diagrams with an identifiable large scale structure. 
In fact, starting from $\mathcal{SL}$ we can add syntactic sugar to handle large scale graphical rewriting.

Given a diagram $g\in \mathcal{L}[n,m]$, its scaled version of size $k$ is a diagram $g_k \in \mathcal{SL}[k \cdot n, k \cdot m]$ inductively defined by: $g_1\df g$ and $\input{gen0.tikz}\df\input{gen1.tikz}$.
 This transformation has been called monoidal multiplexing in \cite{chantawibul2018monoidal}. Notice that these structures require a way to cross wires, here we have a symmetry but a braiding would also work.

\begin{lemma}\label{funk}
	For every $k$, there is a functor $S_k:\textbf{L}\to \mathcal{S}\textbf{L}$, such that $S_k(1)=1\cdot k$ and $S_k(\pi_\mathcal{L}(g))=\pi_\mathcal{SL}(g_k)$.
\end{lemma}

\noindent A proof is given in \ref{pr:funk}. 

\vspace{0.2cm}

These functors ensure that any equation between diagrams still holds at large scale where one application of the scaled rule is in fact hiding $k$ parallel applications of the original rule.

We can go further, if $(g(\alpha):\mathbf{a}\to \mathbf{b})_{\alpha\in A}$ is a family of diagrams indexed by some parameter $\alpha\in A$ then we can index the scaled version by an element $\boldsymbol{\alpha}\in A^k$. This is defined inductively by $g_1(\langle\alpha\rangle)\df g(\alpha)$ and: $\input{genp0.tikz}\df\input{genp1.tikz}$, with $\boldsymbol{\alpha}=\langle \alpha \rangle \boxtimes \boldsymbol{\beta}$ and $\boldsymbol{\beta}\in A^k$.

Notice that  this is how scale spiders are defined in the scalable ZX-calculus \cite{carette2019szx}. We see here that  this general construction applies without difficulties to the other kinds of indexed spiders one can find in ZW or ZH calculi, and to the boxes indexed by numbers in \cite{bonchi2017interacting}. The associated scaled rule, if any, should be \emph{a priori} defined inthe same way. For example, in the ZX-calculus, the phases $\alpha$ and $\beta$
 of two spiders add up when they fuse, so the lists of phases $\boldsymbol{\alpha}$ and $\boldsymbol{\beta}$ add up pointwisely when scaled spiders fuse.

This construction applies naturally to generators but can be applied to any diagrams that we want to see acting at a large scale.

We can refine the global size functor into a functor from $\mathcal{S}\textbf{L}$ to $\textbf{L}$ which forgets dividers and gatherers.

\begin{definition}[Wire stripper]
	Given a monochromatic graphical language $\mathcal{L}$, the \textbf{wire stripper} functor $\{\_\}: \mathcal{S}\textbf{L}\to \textbf{L}$ is defined on colors by $\{\mathbf{a}\}=[\mathbf{a}]$ and on morphisms by $\{\pi_\mathcal{SL}(\delta_k)\}=id_k$, $\{\pi_\mathcal{SL}(\gamma_k)\}=id_k$ and $\{\pi_\mathcal{SL} (x)\}=\pi_\mathcal{L}(x)$ for all generators $x\in \Sigma_{\mathcal{L}}$.
\end{definition}

\begin{lemma}\label{wstrip}
	For every morphism $\omega\in \textbf{L}[n,m]$, we have $\{S_1(\omega)\}=\omega$.
\end{lemma}

\noindent A proof is given in \ref{pr:wstrip}. 

\vspace{0.2cm}

All the properties of $\mathcal{SL}$ follow from a structure theorem which can be seen as an extension to $\mathcal{SL}$ of the rewiring theorem.  

\begin{theorem}[Structure of $\mathcal{S}\textbf{L}$]\label{struct}
	For every $\omega\in \mathcal{SL}[\mathbf{a},\mathbf{b}]$ we have $\mathcal{SL}\vdash \omega=\Gamma_\mathbf{b} \circ S_1(\{\omega\})\circ \Delta_\mathbf{a}$.
\end{theorem}

\noindent A proof is given in \ref{pr:struct}. 

\vspace{0.2cm}

Starting with a graphical language for permutations $\mathbb{P}$ we have $\mathcal{S}\mathbb{P}=\textbf{W}$ and then we recover the rewiring theorem. From this result it follows that the scalable construction enjoys a universal property:

\begin{lemma}[Universal property of $\mathcal{S}\textbf{L}$]\label{univ}
	The following diagram is a pullback square:
	\ctikzfig{pullback}
\end{lemma}

\noindent A proof is given in \ref{pr:univ}. 

\vspace{0.2cm}

The universal property allows to lift interpretation functors.

\begin{lemma}[Scaled interpretation]\label{part}
	Given a monochromatic prop $\textbf{C}$ and an interpretation $\interp{\_}: \textbf{L}\to \textbf{C}$, let $\textbf{C}_p$ the be the symmetric strict monoidal category whose objects are pairs $(n,p_n)$ where $p_n$ is a partition of $n$, and such that $\textbf{C}_p[(n,p_n),(m,p_m)]=\textbf{C}[n,m]$. $\textbf{C}_p$ is an $\N$-colored prop and there is a unique $\N$-colored prop morphism $\interp{\_}_p: \mathcal{S}\textbf{L}\to \textbf{C}_p$ such that $[\_] \circ \interp{\_}_p = \interp{\_}\circ \{\_\}$. We call it the \textbf{scaled interpretation} and it is faithful iff $\interp{\_}$ is faithful.
\end{lemma}

\noindent A proof is given in \ref{pr:part}. 

\vspace{0.2cm}

These results together point out that the scalable construction is just a tool allowing diagrammatical manipulation and is completely orthogonal to the original language. In fact, we even have an equivalence of categories.

\begin{lemma}\label{equiv}
	In $\textbf{SymMonCat}$ we have $\textbf{L}\simeq\mathcal{S}\textbf{L}$.
\end{lemma}

\noindent A proof is given in \ref{pr:equiv}. 

\vspace{0.2cm}

We expect most of the properties of $\textbf{L}$ to be reflected in $\mathcal{S}\textbf{L}$. Here are some specific examples. If $\textbf{L}$ is a dagger category then $\mathcal{S}\textbf{L}$ inherits this structure by setting ${\delta_k}^\dagger \df \gamma_k$. Then we have ${\gamma_k}^\dagger = \delta_k$ and the expansion and elimination equations state that dividers and gatherers are unitary maps.

If $\textbf{L}$ is a compact closed category then so is  $\mathcal{S}\textbf{L}$. Using the scaled version of the cups and caps, we have $\input{bigcup0.tikz}\df \input{bigcup1.tikz}$ and then ${\gamma_k}^t = \delta_k$. However, note that some intuitive topological moves %properties like "rotating" a gatherer into a divider 
do not hold: $\input{trans0.tikz}= \input{trans1.tikz}\neq\input{trans2.tikz}$.

\noindent \textbf{Remark}: Another possibility is to take $\input{abigcup0.tikz}\df \input{abigcup1.tikz}$. Then we recover the topology but we loose the correspondance between equations on simple wires and their scaled version.

\section{The box construction}

In this section, we focus on some fundamental graphical languages. We use the associated completeness results to compress the corresponding diagrams with a box construction.
We consider a way to construct an $\N$-colored graphical language from a monochromatic one: the \textbf{box construction}. The idea is to blackbox a monochromatic prop into a large scale graphical language.

\begin{definition}[$Box_\textbf{P}$]
	  
	Given a monochromatic prop  $\textbf{P}$, let  $\textbf{P}^{\N}$ be the $\N$-colored graphical language defined as 
	$\Sigma_{\textbf{P}^{\N}}[\mathbf{a},\mathbf{b}]\df \begin{cases}U\textbf{P}[n,m] &\text{when } (\mathbf{a},\mathbf{b})=(1\cdot n, 1\cdot m) \\\emptyset&\text{otherwise}\end{cases}$, 
	%\text{ and empelse } \emptyset$, 
	and no equation. For every morphism $f:n\to m$ of $\textbf{P}$ the corresponding  generator in $\textbf{P}^{\N}$  is denoted $\square_f$.
	
	The \textbf{box graphical language} is defined as $Box_\textbf{P}\df \quot{\left(\textbf{P}^{\N}+\mathcal{W}\right)}{\left(Swap, Comp, Tens\right)}$, where swap, Comp, Tens are the following equations:  
	
	\noindent $\bullet$ The \textbf{Swap equation} $swap$ is: $\square_{\sigma_{1,1}}=\gamma_{2} \circ \sigma_{1,1} \circ \delta_{2}$, pictorially: $\input{swa0.tikz}~\stackrel{swap\label{swa}}{=}~\input{swa1.tikz}$.
	
	\noindent $\bullet$ The \textbf{composition equation} $comp_{f,g}$, associated to two morphisms $f:n\to k$ and  $g:k\to m$ of $\textbf{P}$, is  $\square_{g\circ f}=\square_{g} \circ \square_{f}$, pictorially: $\input{com0.tikz}~\stackrel{comp\label{com}}{=}~\input{com1.tikz}$. The set of all composition equations for every $f$ and $g$ is denoted $Comp$.
	
	\noindent $\bullet$ The \textbf{tensor equation} $tens_{f,g}$,  associated to two morphisms  $f:n\to m$ and  $g:k\to l$ of $\textbf{P}$, is $\square_{f\otimes g}= \gamma_{m+l}\circ \left(\square_f \boxtimes  \square_g \right)\circ \delta_{n+k}$, pictorially: $\input{ten0.tikz}\quad\stackrel{tens\label{ten}}{=}\quad\input{ten1.tikz}$. The set of all tensor equations for every $f$ and $g$ is denoted $Tens$.

\end{definition}

The  $\N$-colored prop associated to the \textbf{box graphical language} $Box_\textbf{P}$ is denoted $Box\textbf{P}$. Notice that the box graphical language has one generator for each morphism in $\textbf{P}$. 

From this definition follows directly the existence of a functor $B: \textbf{P}\to Box\textbf{P}$ defined on each morphism $f:n\to m$ by $B(f)=\Delta_m \circ \pi_{Box_\textbf{P}}(\square_f )\circ \Gamma_n$. The equations of $Box_\textbf{P}$ state exactly that $B$ is a symmetric strict monoidal functor. We also have a functor $O: Box\textbf{P}\to \textbf{P}$ defined on generators by $O(\pi_{Box_\textbf{P}}(\square_f))=f$. We have $O\circ B= id_\textbf{P}$.

As with scalable construction, we also have a structure theorem: 

\begin{theorem}[Structure of $Box\textbf{P}$]\label{stbox}
	For each $\omega\in Box(\textbf{P})[\mathbf{a},\mathbf{b}]$ we have $Box(\textbf{P})\vdash \omega=\Gamma_\mathbf{b} \circ B(O(\omega))\circ \Delta_\mathbf{a}$.
\end{theorem}

\noindent A proof is given in \ref{pr:stbox}. 

\vspace{0.2cm}

From this given a monochromatic graphical language $\mathcal{L}$ we can define a functor $Unbox:Box\textbf{L}\to \mathcal{S}\textbf{L}$ by $Unbox(\omega)=\Gamma_\mathbf{b} \circ S_1(O(\omega))\circ \Delta_\mathbf{a}$.

\begin{lemma}\label{same}
	$Unbox$ is an equivalence of categories.
\end{lemma}

\noindent A proof is given in \ref{pr:same}.

\vspace{0.2cm}

So $Box\textbf{L} \simeq \mathcal{S}\textbf{L}$, as a consequence, the two constructions are essentially the same.

We will mostly use the box construction on substructures of $\mathcal{L}$ to obtain box generators inside of $\mathcal{L}$. Of course they are expressible using the usual generators, but they are very useful for compressing diagrams and speed up graphical computations. We now give several examples of various kinds of boxes and some of their interactions with scaled generators.

\subsection{Symmetries and permutations}

We work in the setting of props so the graphical language of permutation is the free graphical language with no signature nor equations. Making an exception,  we describe here the language in the setting of pros in order to start with a familiar example.

\begin{definition}[permutations]
	 The monochromatic graphical language $\mathcal{P}$ has signature: $\input{swap.tikz}$ and equations: $\input{YB0.tikz}=\input{YB1.tikz}$ and $\input{sinv0.tikz}=\input{sinv1.tikz}$.
\end{definition}

%We take for $\mathcal{F}$ the \textbf{permutation matrices}, that is the $\{0,1\}$ square matrices such that there is exactly one $1$ in each row and in each column. This definition holds for any semiring $R$. The permutation matrices of size $n$ are isomorphic to the groups $S_n$ of permutations of a set with $n$ elements. A presentation of this group is given by the transpositions $\tau_i$ which exchange $i$ with $i+1$ and the following equations: for each $1\leq i,j< n$, $\tau_{i}^2=id$, if $|j-i|>1$ then $\tau_{i}\tau_j=\tau_{j}\tau_i$, and if $i\neq n-1$ then $(\tau_{i}\tau_{i+1})^3=id$. The matrix representation of $\tau_i$ is of the form $I_{i-1}\oplus\begin{pmatrix}
%0&1\\
%1&0
%\end{pmatrix}\oplus I_{n-1-i}$, so $\vv{\mathcal{F}}$ is spanned by only one generator: $\vv{\begin{pmatrix}
%	0&1\\
%	1&0
%	\end{pmatrix}}:1\cdot 2 \to 1\cdot 2$ we rename its expanded form $S:2\cdot 1 \to 2\cdot 1$. Now we need to translate the equations. The first equation gives $S^2=id_{2\cdot 1}$, the second equation is immediately true by the exchange law. To ensure the last one we just need to restrict to three wires, this gives: $[(id_1 \otimes S)\circ(S\otimes id_1)]^3=id_{3\cdot 1}$. Using the first equation we can rephrase the last one as: $(S\otimes id_1)\circ(id_1 \otimes S)\circ(S\otimes id_1)=(id_1 \otimes S)\circ(S\otimes id_1)\circ (S\otimes id_1)$. Picturing $S$: \tikzfig{swap} we recover the involutivity of the swap : $\tikzfig{swinv0}=\tikzfig{swinv1}$ and the Yang-Baxter equation : $\tikzfig{YB0}=\tikzfig{YB1}$. Thus any scalable prop contains the permutation matrix arrows.

\begin{lemma}
	The interpretation $\interp{\input{swap.tikz}}= (1,2)(2,1)$ makes $\mathcal{P}$ complete for the prop of permutations $\mathbb{P}$ where each morphism $f:n\to n$ is a permutation of $\{1,\dots,n\}$. Composition is the composition of permutations and the tensor product is the disjoint union.
\end{lemma}

Given any monochromatic prop $\textbf{P}$ there is a unique prop morphism $!:\mathbb{P}\to \textbf{P}$ from the prop of permutations. Thus in any scalable prop we can use permutation boxes from $Box\mathbb{P}$ without ambiguity.

\subsection{Monoid and functions}

\begin{definition}[Commutative monoid]
	The monochromatic graphical language $\mathcal{M}$ has signature: $\left\{\input{mon.tikz},\begin{tikzpicture}
	\begin{pgfonlayer}{nodelayer}
		\node [style=none] (41) at (0, 0) {};
		\node [style=white] (48) at (0, 0) {};
		\node [style=none] (50) at (1, 0) {};
	\end{pgfonlayer}
	\begin{pgfonlayer}{edgelayer}
		\draw [in=180, out=0] (48) to (50.center);
	\end{pgfonlayer}
\end{tikzpicture}
\right\}$ and equations:
	
	\begin{center}
		$\input{assoc0.tikz}=\input{assoc1.tikz}\qquad\input{unit0.tikz}=
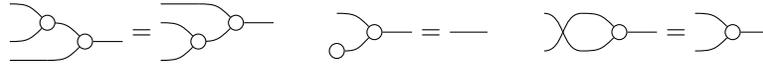
\begin{tikzpicture}
	\begin{pgfonlayer}{nodelayer}
		\node [style=none] (52) at (1, 0) {};
		\node [style=none] (54) at (1, 0) {};
		\node [style=none] (55) at (0, 0) {};
	\end{pgfonlayer}
	\begin{pgfonlayer}{edgelayer}
		\draw [in=0, out=180] (54.center) to (55.center);
	\end{pgfonlayer}
\end{tikzpicture}
\qquad\input{cmt0.tikz}=\input{cmt1.tikz}$
	\end{center}
\end{definition}

\begin{lemma}
	\textbf{Fun} is the prop of functions where each morphism $f:n\to m$ is a function from $\{1,\dots,n\}$ to $\{1,\dots,m\}$. Composition is the composition of functions and the tensor product is the disjoint union. The interpretation where $\interp{\input{mon.tikz}}$ is the unique function $\{1,2\}\to \{1\}$ and $\interp{\begin{tikzpicture}
	\begin{pgfonlayer}{nodelayer}
		\node [style=none] (41) at (0, 0) {};
		\node [style=white] (48) at (0, 0) {};
		\node [style=none] (50) at (1, 0) {};
	\end{pgfonlayer}
	\begin{pgfonlayer}{edgelayer}
		\draw [in=180, out=0] (48) to (50.center);
	\end{pgfonlayer}
\end{tikzpicture}
}$ is the unique function $\emptyset\to \{1\}$, makes $\mathcal{M}$ complete for \textbf{Fun}.
\end{lemma}

From now on, we will depict boxes as arrows to fit the notation of \cite{carette2019szx}: $\input{funarrow.tikz}\df\input{funbox.tikz}$. Every function arrow satisfies: $\input{copyfun0.tikz}=\input{copyfun1.tikz}$ and $\input{erasefun0.tikz}=\begin{tikzpicture}
	\begin{pgfonlayer}{nodelayer}
		\node [style=none] (41) at (0, 0) {};
		\node [style=white] (48) at (0, 0) {};
		\node [style=none] (50) at (1, 0) {};
	\end{pgfonlayer}
	\begin{pgfonlayer}{edgelayer}
		\draw [style=very thick, in=180, out=0] (48) to (50.center);
	\end{pgfonlayer}
\end{tikzpicture}
$.

\subsection{Bialgebras and matrices}

\begin{definition}[Commutative bialgebra]
	The monochromatic graphical language $\mathcal{B}$ is defined as $\mathcal{M}+\mathcal{M}^{op}$ quotiented by the equations: 
	
	\begin{center}
		$\input{bialg0.tikz}=\input{bialg1.tikz}\qquad\input{cp0.tikz}=
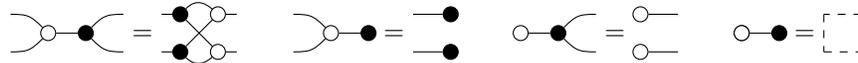
\begin{tikzpicture}
	\begin{pgfonlayer}{nodelayer}
		\node [style=none] (40) at (0, 0.5) {};
		\node [style=black] (41) at (1, -0.5) {};
		\node [style=none] (47) at (0, -0.5) {};
		\node [style=black] (54) at (1, 0.5) {};
	\end{pgfonlayer}
	\begin{pgfonlayer}{edgelayer}
		\draw (47.center) to (41);
		\draw (40.center) to (54);
	\end{pgfonlayer}
\end{tikzpicture}
\qquad\input{cocp0.tikz}=\begin{tikzpicture}
	\begin{pgfonlayer}{nodelayer}
		\node [style=none] (40) at (1, 0.5) {};
		\node [style=white] (41) at (0, -0.5) {};
		\node [style=none] (47) at (1, -0.5) {};
		\node [style=white] (54) at (0, 0.5) {};
	\end{pgfonlayer}
	\begin{pgfonlayer}{edgelayer}
		\draw (47.center) to (41);
		\draw (40.center) to (54);
	\end{pgfonlayer}
\end{tikzpicture}
\qquad\begin{tikzpicture}
	\begin{pgfonlayer}{nodelayer}
		\node [style=black] (41) at (0, 0) {};
		\node [style=white] (48) at (0, 0) {};
		\node [style=black] (50) at (1, 0) {};
	\end{pgfonlayer}
	\begin{pgfonlayer}{edgelayer}
		\draw [in=180, out=0] (48) to (50);
	\end{pgfonlayer}
\end{tikzpicture}
=\input{sca1.tikz}$
	\end{center}

\noindent where the generators of $\mathcal{M}$ are in white and those of $\mathcal{M}^{op}$ in black.
\end{definition}

\begin{lemma}\cite{pirashvili2002prop}
	$\mathcal{M}(\mathbb{N})$ is the prop of integer matrices where each morphism $f:n\to m$ is a matrix in $\mathcal{M}_{m\times n}(\mathbb{N})$. Composition is the matrix product and the tensor product is the direct sum. The interpretation $\interp{\input{mon.tikz}}= \begin{pmatrix}
	1&1
	\end{pmatrix}$, $\interp{\begin{tikzpicture}
	\begin{pgfonlayer}{nodelayer}
		\node [style=none] (41) at (0, 0) {};
		\node [style=white] (48) at (0, 0) {};
		\node [style=none] (50) at (1, 0) {};
	\end{pgfonlayer}
	\begin{pgfonlayer}{edgelayer}
		\draw [in=180, out=0] (48) to (50.center);
	\end{pgfonlayer}
\end{tikzpicture}
}= ()\in \mathcal{M}_{1\times 0}(\mathbb{N})$,$\interp{\input{bcomon.tikz}}= \begin{pmatrix}
	1\\ 1
	\end{pmatrix}$ and $\interp{\begin{tikzpicture}
	\begin{pgfonlayer}{nodelayer}
		\node [style=none] (41) at (1, 0) {};
		\node [style=black] (48) at (1, 0) {};
		\node [style=none] (50) at (0, 0) {};
	\end{pgfonlayer}
	\begin{pgfonlayer}{edgelayer}
		\draw [in=0, out=180] (48) to (50.center);
	\end{pgfonlayer}
\end{tikzpicture}
}= ()\in \mathcal{M}_{0\times 1}(\mathbb{N})$ makes $\mathcal{B}$  complete for $\mathcal{M}(\mathbb{N})$.
\end{lemma}

A matrix arrow indexed by $A:n\to m$ corresponds to a bipartite multigraph between $n$ black vertices and $m$ white vertices. $A$ is nothing but the biadjacency matrix of this multigraph. Thus, in the presence of bialgebra, the box construction allows to compress a bipartite sub-diagram into a single matrix arrow.  The properties of matrix arrows generalize the ones of function boxes. The following equations hold for any matrix $A$:

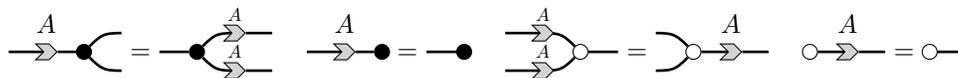
\begin{center}
	\begin{tabular}{cccc}
		$\input{copymat0.tikz}=\input{copymat1.tikz}$ & $\input{coerasemat0.tikz}=\begin{tikzpicture}
	\begin{pgfonlayer}{nodelayer}
		\node [style=none] (41) at (1, 0) {};
		\node [style=black] (48) at (1, 0) {};
		\node [style=none] (50) at (0, 0) {};
	\end{pgfonlayer}
	\begin{pgfonlayer}{edgelayer}
		\draw [style=very thick, in=0, out=180] (48) to (50.center);
	\end{pgfonlayer}
\end{tikzpicture}
$&
		$\input{cocopymat0.tikz}=\input{cocopymat1.tikz}$ & $\input{erasemat0.tikz}=\begin{tikzpicture}
	\begin{pgfonlayer}{nodelayer}
		\node [style=none] (41) at (0, 0) {};
		\node [style=white] (48) at (0, 0) {};
		\node [style=none] (50) at (1, 0) {};
	\end{pgfonlayer}
	\begin{pgfonlayer}{edgelayer}
		\draw [style=very thick, in=180, out=0] (48) to (50.center);
	\end{pgfonlayer}
\end{tikzpicture}
$
	\end{tabular}
\end{center}

Furthermore: $\input{addmat0.tikz}=\input{addmat1.tikz}$.

If we add to $\mathcal{B}$ the generator $\input{ant.tikz}$ and the equation $\input{hopf0.tikz}=\begin{tikzpicture}
	\begin{pgfonlayer}{nodelayer}
		\node [style=black] (57) at (1, 0) {};
		\node [style=white] (58) at (2, 0) {};
		\node [style=none] (59) at (0, 0) {};
		\node [style=none] (60) at (3, 0) {};
	\end{pgfonlayer}
	\begin{pgfonlayer}{edgelayer}
		\draw (59.center) to (57);
		\draw (58) to (60.center);
	\end{pgfonlayer}
\end{tikzpicture}
$ we obtain the graphical language $\mathcal{H}$ of Hopf algebras.
$\mathcal{H}$ is complete for matrices over $\mathbb{Z}$ with $\input{mmat0.tikz}=\input{mmat1.tikz}$.

\subsection{Interacting Hopf algebras and linear relations}

\begin{definition}[Interacting Hopf algebras]
	The monochromatic graphical language $\mathcal{IH}$ is defined as $\mathcal{B}+\mathcal{B}^{op}$ quotiented by the equations: %with one more generator equations:
	\begin{center}
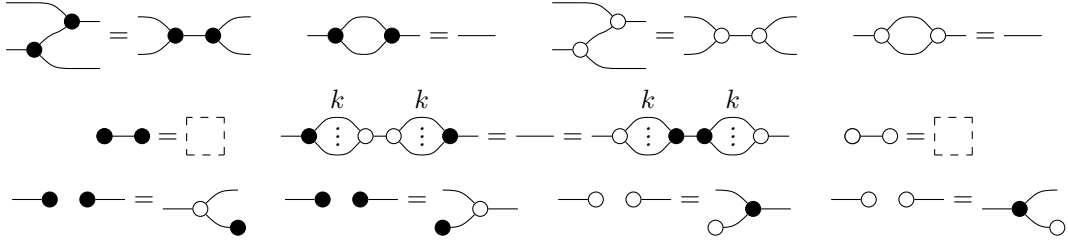
		\begin{tabular}{cccc}
			$\input{bfrob0.tikz}=\input{bfrob1.tikz}$ &$\input{bspe0.tikz}=\begin{tikzpicture}
	\begin{pgfonlayer}{nodelayer}
		\node [style=none] (57) at (1, 0) {};
		\node [style=none] (59) at (0, 0) {};
	\end{pgfonlayer}
	\begin{pgfonlayer}{edgelayer}
		\draw (59.center) to (57.center);
	\end{pgfonlayer}
\end{tikzpicture}
$& $\input{wfrob0.tikz}=\input{wfrob1.tikz}$&$\input{wspe.tikz}=\begin{tikzpicture}
	\begin{pgfonlayer}{nodelayer}
		\node [style=none] (57) at (1, 0) {};
		\node [style=none] (59) at (0, 0) {};
	\end{pgfonlayer}
	\begin{pgfonlayer}{edgelayer}
		\draw (59.center) to (57.center);
	\end{pgfonlayer}
\end{tikzpicture}
$\\[0.5cm]
			\multicolumn{4}{c}{$\begin{tikzpicture}
	\begin{pgfonlayer}{nodelayer}
		\node [style=black] (41) at (0, 0) {};
		\node [style=black] (48) at (0, 0) {};
		\node [style=black] (50) at (1, 0) {};
	\end{pgfonlayer}
	\begin{pgfonlayer}{edgelayer}
		\draw [in=180, out=0] (48) to (50);
	\end{pgfonlayer}
\end{tikzpicture}
=\input{sca1.tikz}\qquad\input{num0.tikz}=\begin{tikzpicture}
	\begin{pgfonlayer}{nodelayer}
		\node [style=none] (57) at (1, 0) {};
		\node [style=none] (59) at (0, 0) {};
	\end{pgfonlayer}
	\begin{pgfonlayer}{edgelayer}
		\draw (59.center) to (57.center);
	\end{pgfonlayer}
\end{tikzpicture}
=\input{num1.tikz}\qquad\begin{tikzpicture}
	\begin{pgfonlayer}{nodelayer}
		\node [style=black] (41) at (0, 0) {};
		\node [style=white] (48) at (0, 0) {};
		\node [style=white] (50) at (1, 0) {};
	\end{pgfonlayer}
	\begin{pgfonlayer}{edgelayer}
		\draw [in=180, out=0] (48) to (50);
	\end{pgfonlayer}
\end{tikzpicture}
=\input{sca1.tikz}$}\\[0.5cm]
			&&&\\[-0.5cm]
			$\begin{tikzpicture}
	\begin{pgfonlayer}{nodelayer}
		\node [style=none] (40) at (0, 0) {};
		\node [style=none] (41) at (1, 0) {};
		\node [style=none] (47) at (3, 0) {};
		\node [style=black] (48) at (1, 0) {};
		\node [style=black] (50) at (2, 0) {};
	\end{pgfonlayer}
	\begin{pgfonlayer}{edgelayer}
		\draw (48) to (40.center);
		\draw (50) to (47.center);
	\end{pgfonlayer}
\end{tikzpicture}
=\input{int5.tikz}$ & $\begin{tikzpicture}
	\begin{pgfonlayer}{nodelayer}
		\node [style=none] (40) at (0, 0) {};
		\node [style=none] (41) at (1, 0) {};
		\node [style=none] (47) at (3, 0) {};
		\node [style=black] (48) at (1, 0) {};
		\node [style=black] (50) at (2, 0) {};
	\end{pgfonlayer}
	\begin{pgfonlayer}{edgelayer}
		\draw (48) to (40.center);
		\draw (50) to (47.center);
	\end{pgfonlayer}
\end{tikzpicture}
=\input{int7.tikz}$&  $\begin{tikzpicture}
	\begin{pgfonlayer}{nodelayer}
		\node [style=none] (40) at (0, 0) {};
		\node [style=none] (41) at (1, 0) {};
		\node [style=none] (47) at (3, 0) {};
		\node [style=white] (48) at (1, 0) {};
		\node [style=white] (50) at (2, 0) {};
	\end{pgfonlayer}
	\begin{pgfonlayer}{edgelayer}
		\draw (48) to (40.center);
		\draw (50) to (47.center);
	\end{pgfonlayer}
\end{tikzpicture}
=\input{int1.tikz}$&$\begin{tikzpicture}
	\begin{pgfonlayer}{nodelayer}
		\node [style=none] (40) at (0, 0) {};
		\node [style=none] (41) at (1, 0) {};
		\node [style=none] (47) at (3, 0) {};
		\node [style=white] (48) at (1, 0) {};
		\node [style=white] (50) at (2, 0) {};
	\end{pgfonlayer}
	\begin{pgfonlayer}{edgelayer}
		\draw (48) to (40.center);
		\draw (50) to (47.center);
	\end{pgfonlayer}
\end{tikzpicture}
=\input{int3.tikz}$
		\end{tabular}
	\end{center}
	
	The generators of $\mathcal{B}$ are depicted as before and those of $\mathcal{B}^{op}$ are depicted by exchanging black and white. The number $k$ of parallel wires must be at least $2$.
	
\end{definition}

\begin{lemma}\cite{bonchi2017interacting} 
	$\textbf{LinRel}$ is the prop of linear relation where each morphism $f:n\to m$ is a linear subspace of $\mathbb{Q}^{n+m}$. Composition is composition of relation and the tensor product is direct sum. The interpretation: $\interp{\input{mon.tikz}}= \{(x,y,z), x+y=z\}$, $\interp{\begin{tikzpicture}
	\begin{pgfonlayer}{nodelayer}
		\node [style=none] (41) at (0, 0) {};
		\node [style=white] (48) at (0, 0) {};
		\node [style=none] (50) at (1, 0) {};
	\end{pgfonlayer}
	\begin{pgfonlayer}{edgelayer}
		\draw [in=180, out=0] (48) to (50.center);
	\end{pgfonlayer}
\end{tikzpicture}
}= \{0\}\times\{0\}$, $\interp{\input{comon.tikz}}= \{(x,y,z), x=y+z\}$, $\interp{\begin{tikzpicture}
	\begin{pgfonlayer}{nodelayer}
		\node [style=none] (41) at (1, 0) {};
		\node [style=white] (48) at (1, 0) {};
		\node [style=none] (50) at (0, 0) {};
	\end{pgfonlayer}
	\begin{pgfonlayer}{edgelayer}
		\draw [in=0, out=180] (48) to (50.center);
	\end{pgfonlayer}
\end{tikzpicture}
}= \{0\}\times\{0\}$, $\interp{\input{bmon.tikz}}= \{(x,x,x), x\in\mathbb{Q}\}$, $\interp{\begin{tikzpicture}
	\begin{pgfonlayer}{nodelayer}
		\node [style=none] (41) at (1, 0) {};
		\node [style=black] (48) at (1, 0) {};
		\node [style=none] (50) at (0, 0) {};
	\end{pgfonlayer}
	\begin{pgfonlayer}{edgelayer}
		\draw [in=0, out=180] (48) to (50.center);
	\end{pgfonlayer}
\end{tikzpicture}
}= \{0\}\times\mathbb{Q}$, $\interp{\input{bcomon.tikz}}= \{(x,x,x), x\in\mathbb{Q}\}$ and  $\interp{\begin{tikzpicture}
	\begin{pgfonlayer}{nodelayer}
		\node [style=none] (41) at (1, 0) {};
		\node [style=black] (48) at (1, 0) {};
		\node [style=none] (50) at (0, 0) {};
	\end{pgfonlayer}
	\begin{pgfonlayer}{edgelayer}
		\draw [in=0, out=180] (48) to (50.center);
	\end{pgfonlayer}
\end{tikzpicture}
}= \mathbb{Q}\times\{0\}$   makes $\mathcal{IH}$ complete for $\textbf{LinRel}$.

\end{lemma}

We can interpret matrices in $\mathcal{IH}$ and thus have matrix arrows. Moreover we have a compact structure on simple wires given by $\input{bcup.tikz}$ and $\input{bcap.tikz}$. We define backward matrix as $\input{barrow.tikz}\df\input{barrow0.tikz}$.

Backward matrix arrows have the same properties as matrix arrows but in the reverse direction. We also have:

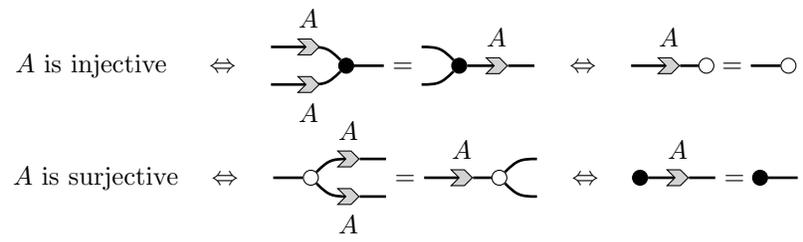
\begin{center}
	\begin{tabular}{c}
		$A\text{ is injective } \quad\Leftrightarrow\quad \input{injcopy.tikz}=\input{injcopy1.tikz}\quad \Leftrightarrow \quad\input{injerase0.tikz}=\begin{tikzpicture}
	\begin{pgfonlayer}{nodelayer}
		\node [style=white] (50) at (1, 0) {};
		\node [style=none] (52) at (0, 0) {};
	\end{pgfonlayer}
	\begin{pgfonlayer}{edgelayer}
		\draw [style=very thick] (52.center) to (50);
	\end{pgfonlayer}
\end{tikzpicture}
$\\
		$A\text{ is surjective} \quad\Leftrightarrow\quad\input{surjcopy0.tikz}=\input{surjcopy1.tikz} \quad\Leftrightarrow\quad\input{surjerase0.tikz}=\begin{tikzpicture}
	\begin{pgfonlayer}{nodelayer}
		\node [style=black] (50) at (0, 0) {};
		\node [style=none] (52) at (1, 0) {};
	\end{pgfonlayer}
	\begin{pgfonlayer}{edgelayer}
		\draw [style=very thick] (52.center) to (50);
	\end{pgfonlayer}
\end{tikzpicture}
$
	\end{tabular}
\end{center}

In practice we will not use linear relation boxes but only matrix arrows since any linear relation factorizes into matrix arrows. In fact, all the equations and properties given in this section so far can be summed up in the powerful formula from \cite{zanasi2018interacting}:

\begin{center}
\begin{tabular}{|c|}
	\hline \\
	$\input{span.tikz}=\input{cospan.tikz} \quad\Leftrightarrow\quad Im\begin{pmatrix}
	C\\ D
	\end{pmatrix}=Ker\begin{pmatrix}
	A& - B 
	\end{pmatrix}$\\[0.5cm]
	\hline
\end{tabular}	
\end{center}

\section{Examples in graphical languages}

In practice, when we identify a substructure in a graphical language we can use the scalable technics to manipulate boxes. However, in general those are not free and we can have extra equations between boxes. Moreover, sometimes the structures are a little different than expected and thus some adjustments need to be made on the properties of boxes.
We illustrate this by considering three concrete examples of bialgebras from categorical quantum mechanics. They appear respectively in the ZH, ZW and ZX-calculus. Notice that we use the definitions of these bialgebras given in \cite{recipe} which are essentially equivalent, but may slightly differ from the original ones. %, more details on this are in \cite{recipe}. 
Each of the three graphical languages has an interpretation which make it complete for the prop of qubits where each morphism $f:n\to m$ is a matrix in $\mathcal{M}_{2^m,2^n}(\mathbb{C})$. The composition is the matrix product and the tensor is the tensor product of matrices. A basis of $\mathbb{C}^{2^n}$ is denoted by the $|x\rangle$ where $x$ is a binary word of size $n$. $e_i$ is the binary word with $0$ everywhere except in the $i$th coordinate where it is $1$.

We recall that a bialgebra corresponds to matrices over the semiring $(\mathbb{N},+,\times)$. The three bialgebras share the same comonoid defined by $\interp{\input{bcomon.tikz}}= \begin{pmatrix}
	1& 0\\ 0&0 \\ 0&0 \\ 0&1
	\end{pmatrix}$ and $\interp{\begin{tikzpicture}
	\begin{pgfonlayer}{nodelayer}
		\node [style=none] (41) at (1, 0) {};
		\node [style=black] (48) at (1, 0) {};
		\node [style=none] (50) at (0, 0) {};
	\end{pgfonlayer}
	\begin{pgfonlayer}{edgelayer}
		\draw [in=0, out=180] (48) to (50.center);
	\end{pgfonlayer}
\end{tikzpicture}
}= \begin{pmatrix}
	1&1
	\end{pmatrix}$.
	
\vspace{0.3cm}
The \textbf{ZH bialgebra} is defined by: $\interp{\input{mon.tikz}}= \begin{pmatrix}
	1&1&1&0\\0&0&0&1
	\end{pmatrix}$ and $\interp{\begin{tikzpicture}
	\begin{pgfonlayer}{nodelayer}
		\node [style=none] (41) at (0, 0) {};
		\node [style=white] (48) at (0, 0) {};
		\node [style=none] (50) at (1, 0) {};
	\end{pgfonlayer}
	\begin{pgfonlayer}{edgelayer}
		\draw [in=180, out=0] (48) to (50.center);
	\end{pgfonlayer}
\end{tikzpicture}
}= \begin{pmatrix}
	1\\0
	\end{pmatrix}$.

We have $\input{two.tikz}=\begin{tikzpicture}
	\begin{pgfonlayer}{nodelayer}
		\node [style=none] (57) at (0.75, 0) {};
		\node [style=none] (59) at (0, 0) {};
	\end{pgfonlayer}
	\begin{pgfonlayer}{edgelayer}
		\draw (59.center) to (57.center);
	\end{pgfonlayer}
\end{tikzpicture}
$, which essentially means ``2=1''. As a consequence, we are working with matrices over the semiring $\quot{(\mathbb{N},+,\times)}{(2=1)}$. This is exactly the boolean semiring $\left(\lor,\land\right)$. There is no more quotienting since $A|e_j\rangle= |A_j\rangle $ where $A_j$ is the $j$th column of $A$. All $\{0,1\}$-matrix arrows are different.

\vspace{0.3cm}

The \textbf{ZW bialgebra} is defined by $\interp{\input{mon.tikz}}= \begin{pmatrix}
	1&0&0&0\\0&1&1&0
	\end{pmatrix}$ and $\interp{\begin{tikzpicture}
	\begin{pgfonlayer}{nodelayer}
		\node [style=none] (41) at (0, 0) {};
		\node [style=white] (48) at (0, 0) {};
		\node [style=none] (50) at (1, 0) {};
	\end{pgfonlayer}
	\begin{pgfonlayer}{edgelayer}
		\draw [in=180, out=0] (48) to (50.center);
	\end{pgfonlayer}
\end{tikzpicture}
}= \begin{pmatrix}
	1\\0
	\end{pmatrix}$.

We have $\input{three.tikz}=\input{two.tikz}$. Thus,  we are working with matrices over the semiring $\quot{(\mathbb{N},+,\times)}{(2=3)}$. However, not all $\{0,1,2\}$-matrices have a distinct interpretation. If $A$ is a $\{0,1\}$-matrix then $A|e_j\rangle= |A_j\rangle $, but if $A$ has a $2$ in the $j$th column then for all $x$ such that $x_j =1$ we have $A|x\rangle=0$. We cannot distinguish the coefficients in a column with a $2$. So the arrows corresponds to matrices with only $\{0,1\}$-coefficients except in some columns which are full of $2$s.

\vspace{0.3cm}

The \textbf{ZX bialgebra} is defined by $\interp{\input{mon.tikz}}= \begin{pmatrix}
	1&0&0&1\\0&1&1&0
	\end{pmatrix}$ and $\interp{\begin{tikzpicture}
	\begin{pgfonlayer}{nodelayer}
		\node [style=none] (41) at (0, 0) {};
		\node [style=white] (48) at (0, 0) {};
		\node [style=none] (50) at (1, 0) {};
	\end{pgfonlayer}
	\begin{pgfonlayer}{edgelayer}
		\draw [in=180, out=0] (48) to (50.center);
	\end{pgfonlayer}
\end{tikzpicture}
}= \begin{pmatrix}
	1\\0
	\end{pmatrix}$.

We have $\input{two.tikz}=\begin{tikzpicture}
	\begin{pgfonlayer}{nodelayer}
		\node [style=white] (57) at (1, 0) {};
		\node [style=none] (59) at (0, 0) {};
		\node [style=none] (60) at (3, 0) {};
		\node [style=black] (61) at (2, 0) {};
	\end{pgfonlayer}
	\begin{pgfonlayer}{edgelayer}
		\draw (59.center) to (57);
		\draw (61) to (60.center);
	\end{pgfonlayer}
\end{tikzpicture}
$. We are working with matrices over the semiring $\quot{(\mathbb{N},+,\times)}{(2=0)}$. This is the field $\mathbb{F}_2$. There is no more quotienting since $A|e_j\rangle= |A_j\rangle $. All $\{0,1\}$-matrices have a distinct interpretation. As the presence of a field suggests we can extend ZX into a \textbf{scaled} interacting Hopf algebras. This means the equational theory of $\mathcal{IH}$ holds up to some scalars. We can still work with linear relations over $\mathbb{F}_2$ but the main equations needs to be renormalized as follow:

\begin{center}
	
		$\input{span.tikz}=\input{zxcospan.tikz} \quad\Leftrightarrow\quad Im\begin{pmatrix}
		C\\ D
		\end{pmatrix}=Ker\begin{pmatrix}
		A& B 
		\end{pmatrix}$
\end{center}

Where $k\df dim\left(Ker\begin{pmatrix}C\\D\end{pmatrix}\right)$ and $\interp{\bigstar}=\frac{1}{\sqrt{2}}$. This equation subsumes most of the properties of the $\mathcal{S}ZX$-calculus of \cite{carette2019szx}.

\clearpage

%%
%% Bibliography
%%

%% Please use bibtex, 

\bibliography{scal}

\appendix

\section{Proofs}

\subsection{Proofs for section 1}

\begin{proof}[Proof of Lemma \ref{glcat}]\phantomsection\label{glcat:pr}
	We will use the string diagrams notation for natural transformations. See \cite{curien2008joy}.
	A translation is pictured: \input{translat.tikz}. The soundness condition states that there is a prop morphism $Bold(\theta)$ such that: $\input{condb0.tikz}=\input{condb1.tikz}$.
	The composition between to translations $\theta:\mathcal{L}\to \mathcal{Y}$ and $\gamma:\mathcal{Y}\to \mathcal{K}$ is defined as $\gamma\theta\df\mu_{\Sigma_\mathcal{K}}\circ T\gamma\circ \theta$, pictorially: 
	
	\begin{center}
		$\input{transcomp0.tikz}\df\input{transcomp1.tikz}$
	\end{center}

	We need to check the soundness condition for the composition. We have:
	
	\begin{center}
		$\input{transcompcond0.tikz}=\input{transcompcond1.tikz}=\input{transcompcond2.tikz}$
	\end{center}
	
	So the soundness condition is satisfied with $Bold(\gamma\theta)=Bold(\gamma)\circ Bold(\theta)$.
	
	Given three translations $\theta:\mathcal{L}\to \mathcal{Y}$, $\gamma:\mathcal{Y}\to \mathcal{K}$ and  $\omega:\mathcal{K}\to \mathcal{M}$ we have: 
	
	\begin{center}
		$\omega(\gamma\theta)=\input{trassoc0.tikz}=\input{trassoc1.tikz}=(\omega\gamma)\theta$
	\end{center}

	So the composition is associative.
	
	%On one side we have $(\omega\gamma)\theta=\mu_{\Sigma_\mathcal{M}}\circ T\omega\gamma\circ \theta=\mu_{\Sigma_\mathcal{M}}\circ T(\mu_{\Sigma_\mathcal{M}}\circ T\omega\circ \gamma)\circ \theta=\mu_{\Sigma_\mathcal{M}}\circ T\mu_{\Sigma_\mathcal{M}}\circ T^2 \omega\circ T\gamma \circ \theta$
	
	%and on the other: $\omega(\gamma\theta)=\mu_{\Sigma_\mathcal{M}}\circ T\omega\circ \gamma\theta=\mu_{\Sigma_\mathcal{M}}\circ T\omega\circ (\mu_{\Sigma_\mathcal{K}}\circ T\gamma\circ \theta)=\mu_{\Sigma_\mathcal{M}}\circ T\omega\circ \mu_{\Sigma_\mathcal{K}}\circ T\gamma\circ \theta=\mu_{\Sigma_\mathcal{M}}\circ \mu_{T\Sigma_\mathcal{M}}\circ T^2\omega\circ T\gamma\circ \theta=\mu_{\Sigma_\mathcal{M}}\circ T\mu_{\Sigma_\mathcal{M}}\circ T^2 \omega\circ T\gamma \circ \theta$.
	
	%The last two equalities are the naturality and the associativity of $\mu$. So $(\omega\gamma)\theta=\omega(\gamma\theta)$, the composition is associative.
	
	The unit is defined as $id_\mathcal{Y}\df \eta_{\Sigma_\mathcal{Y}}$, pictorially: 
	
	\begin{center}
		$\input{transid0.tikz}\df\input{transid1.tikz}$
	\end{center}
	
	The soundness condition is satisfied: $\input{condid0.tikz}=\input{condid1.tikz}$ and we have $Bold(id_\mathcal{Y})=id_\textbf{Y}$.
	%We have $\omega ~ id_\mathcal{K}=\mu_{\Sigma_\mathcal{M}}\circ T\omega\circ \eta_{\Sigma_\mathcal{K}}=\mu_{\Sigma_\mathcal{M}}\circ \eta_{T\Sigma_\mathcal{M}} \circ \omega=\omega$ and $id_\mathcal{K} ~ \gamma =\mu_{\Sigma_\mathcal{K}}\circ T\eta_{\Sigma_\mathcal{K}}\circ \gamma=\gamma$.
	
	We have: \begin{center}
		$id_\mathcal{Y} ~ \theta=\input{transidp0.tikz}=\input{transidp1.tikz}=\theta$
	\end{center} and \begin{center}
	$\gamma ~ id_\mathcal{Y}=\input{transidp2.tikz}=\input{transidp3.tikz}=\gamma$
\end{center}
	
	So $\textbf{GL}$ is a category and setting $Bold(\mathcal{L})=\textbf{L}$, we have a functor $Bold: \textbf{GL}\to \text{C-}\textbf{Prop}$.
	
	Since $C$-\textbf{Prop} is equivalent to the Eilenberg-Moore category of $T$ then each prop $\textbf{P}$ is a coequalizer of free  props \cite{barr2005toposes}. An explicit coequalizer is:
	
	\begin{center}
		\input{coeql.tikz}
	\end{center}
	
	In fact graphically: $\input{pcoeq0.tikz}=\input{pcoeq1.tikz}$.
	
	Let's take $\Sigma_\mathcal{P}\df U \textbf{P}$, $E_\mathcal{P}\df UFU\textbf{P}$, $r_\mathcal{P}\df id_{FUFU\textbf{P}}$ and $l_\mathcal{P}\df F\eta_{U\epsilon_{\textbf{P}}}$.
	We have: $\epsilon_{FU\textbf{P}}\circ l_\mathcal{P}=\epsilon_{FU\textbf{P}}\circ F\eta_{U\epsilon_{\textbf{P}}}=\input{fcoeq0.tikz}=\input{fcoeq1.tikz}=FU\epsilon_{\textbf{P}}$. Thus we can define $\mathcal{P}\df \left(\Sigma_\mathcal{P},E_\mathcal{P},l_\mathcal{P},r_\mathcal{P}\right)$ and we have $Bold(\mathcal{P})\simeq \textbf{P}$. So $Bold$ is essentially surjective.
	
	Let $s$ be a section of $\pi_\mathcal{Y}$. Given a prop morphism $f:Bold(\mathcal{L})\to Bold(\mathcal{Y})$ we define a translation $\theta_f:\mathcal{L}\to \mathcal{Y}$ by $\theta_f\df Usf\pi_\mathcal{L}\circ \eta_{\Sigma_\mathcal{L}}$. We check the soundness condition: $ \pi_\mathcal{Y}\circ \left(\epsilon_{F\Sigma_\mathcal{Y}}\circ F\theta_f\right)=\input{fcoeq2.tikz}=\input{fcoeq3.tikz}=\input{fcoeq4.tikz}=f\circ \pi_\mathcal{L}$. This gives us $Bold(\theta_f)=f$. So $Bold$ is full.
\end{proof}

\begin{proof}[Proof of Lemma \ref{copr}]\phantomsection\label{copr:pr}
	Since coproducts commute with coequalizers we have the following coequalizer:
	
	\begin{center}
		\input{coeqsum.tikz}
	\end{center}
	
	Furthermore: $F\iota_{\Sigma_\mathcal{L}}\circ \epsilon_{F\Sigma_\mathcal{L}}\circ Fl_\mathcal{L}= \epsilon_{F\Sigma_{\mathcal{L}+\mathcal{Y}}}\circ FUF\iota_{\Sigma_\mathcal{L}}\circ Fl_\mathcal{L}$.
	
	And so: $\epsilon_{F\Sigma_{\mathcal{L}+\mathcal{Y}}}\circ Fl_{\mathcal{L}+\mathcal{Y}}=\epsilon_{F\Sigma_{\mathcal{L}+\mathcal{Y}}}\circ F[UF(\iota_{\Sigma_\mathcal{L}})\circ l_\mathcal{L},UF(\iota_{\Sigma_\mathcal{L}})\circ l_\mathcal{Y}]=[\epsilon_{F\Sigma_{\mathcal{L}+\mathcal{Y}}}\circ FUF(\iota_{\Sigma_\mathcal{L}})\circ l_\mathcal{L},\epsilon_{F\Sigma_{\mathcal{L}+\mathcal{Y}}}\circ FUF(\iota_{\Sigma_\mathcal{L}})\circ l_\mathcal{Y}]=[F\iota_{\Sigma_\mathcal{L}}\circ \epsilon_{F\Sigma_\mathcal{L}}\circ Fl_\mathcal{L},F\iota_{\Sigma_\mathcal{Y}}\circ \epsilon_{F\Sigma_\mathcal{Y}}\circ Fl_\mathcal{Y}]=(\epsilon_{F\Sigma_\mathcal{L}}\circ Fl_\mathcal{L})+(\epsilon_{F\Sigma_\mathcal{Y}}\circ Fl_\mathcal{Y})$.
	
\end{proof}

\begin{proof}[Proof of Lemma \ref{coeq}]\phantomsection\label{coeq:pr}
	We use the following diagram:
	
	\begin{center}
		\input{qcoeq.tikz}
	\end{center}
	
	First we have $\pi_{\quot{\mathcal{L}}{E}}\circ \epsilon_{F\Sigma_{\mathcal{L}}}\circ Fl_{\mathcal{L}}=\pi_{\quot{\mathcal{L}}{E}}\circ \epsilon_{F\Sigma_{\mathcal{L}}}\circ Fr_{\mathcal{L}}$ and thus there is a unique $\pi:\textbf{L}\to Bold(\quot{\mathcal{L}}{E})$ such that $\pi\circ \pi_\mathcal{L}=\pi_{\quot{\mathcal{L}}{E}}$. Then: $\pi\circ \pi_\mathcal{L}\circ\epsilon_{F\Sigma_{\mathcal{L}}}\circ Fl=\pi\circ \pi_\mathcal{L}\circ \epsilon_{F\Sigma_{\mathcal{L}}}\circ Fr$. 
	
	It remains to show the universal property. Let $f:\textbf{L}\to \textbf{K}$ be a prop morphism such that: $f\circ \pi_\mathcal{L}\circ \epsilon_{F\Sigma_{\mathcal{L}}}\circ Fl=f\circ \pi_\mathcal{L}\circ \epsilon_{F\Sigma_{\mathcal{L}}}\circ Fr$.
	
	We also have $f\circ \pi_\mathcal{L}\circ  \epsilon_{F\Sigma_{\mathcal{L}}}\circ Fl_{\mathcal{L}}=f\circ \pi_\mathcal{L}\circ  \epsilon_{F\Sigma_{\mathcal{L}}}\circ Fr_{\mathcal{L}}$ and then the universal property of the coproduct gives us $f\circ \pi_\mathcal{L}\circ \epsilon_{ F \Sigma_{\mathcal{L} }}\circ F[l_\mathcal{L} ,l]=f\circ \pi_\mathcal{L}\circ \epsilon_{F\Sigma_{\mathcal{L}}} \circ F[r_\mathcal{L} ,r]$.
	
	So there is a unique $g:Bold(\quot{\mathcal{L}}{E})\to \textbf{K}$ such that: $f\circ \pi_\mathcal{L}=g\circ \pi_{\quot{\mathcal{L}}{E}}=g\circ \pi \circ \pi_\mathcal{L}$. And since $\pi_\mathcal{L}$ is an epimorphism: $f=g\circ \pi$.
\end{proof}

\subsection{Proofs for section 2}

\begin{proof}[Proof of Lemma \ref{thm:rw}]\phantomsection\label{pr:rw}
	
	We will show that for any diagram $\omega:a\to b$ we have $\mathbb{W}\vdash \omega=\Gamma_b \circ \sigma \circ \Delta_a$ where $\sigma$ is a permutation. To do so we define, for each wire in the diagram, its \textbf{situation} as a couple of elements of $\{i,o,d,g\}$. The situation of a wire describes what the wire is linked to what: $i$ stands for input, $o$ for output, $d$ for non identity divider, and $g$ for non identity gatherer. For example, a wire which links an input to an output has situation $(i,o)$ and a wire linking a gatherer to a divider has situation $(g,d)$. The possible situations for a simple wire are: $(i,o)$, $(i,g)$, $(d,o)$, and $(d,g)$. The possible situations for a big wire are the same plus $(i,d)$, $(d,d)$, $(g,o)$, $(g,g)$ and $(g,d)$.
	
	We say that a diagram is \textbf{expanded} if it contains no big wire in one of the \textbf{bad} situations which are $(g,d)$ and $(d,g)$.
	
	The expanded conditions enforce a unique structure. In fact the only expanded diagrams are exactly the $\Gamma_b\circ\sigma\circ\Delta_a$. Thus it only remains to show that any diagram can be rewritten into an expanded one.
	
	We proceed by induction on the size of the biggest big wire in a bad situation.
	
	If there are no such wire then we are already in expanded form. Else, we consider the biggest wires in a bad situations. If a wire is in situation $(g,d)$ then the elimination rule \ref{elim} can be applied and decreases strictly the size of the big wires in bad situations. If a wire is in situation $(g,d)$ then the expantion rule \ref{exp} can be applied and decreases strictly the size of the big wires in bad situations. Thus all the wire in bad situations can be removed and replace by wire with strictly smaller size. Then by induction all diagrams can be rewritten into expanded form. 
	
\end{proof}

\begin{proof}[Proof of Lemma \ref{funk}]\phantomsection\label{pr:funk}
	By induction. $S_1$ is clearly a functor. The expansion equation ensures that $S_k(id_n)=id_{1\cdot n}$: $\input{sid0.tikz}=\input{sid1.tikz}$. The elimination equation ensures that $S_k(g\circ h)=S_k(g)\circ S_k(h)$:
	$\input{scp0.tikz}=\input{scp1.tikz}=\input{scp2.tikz}$.
\end{proof}

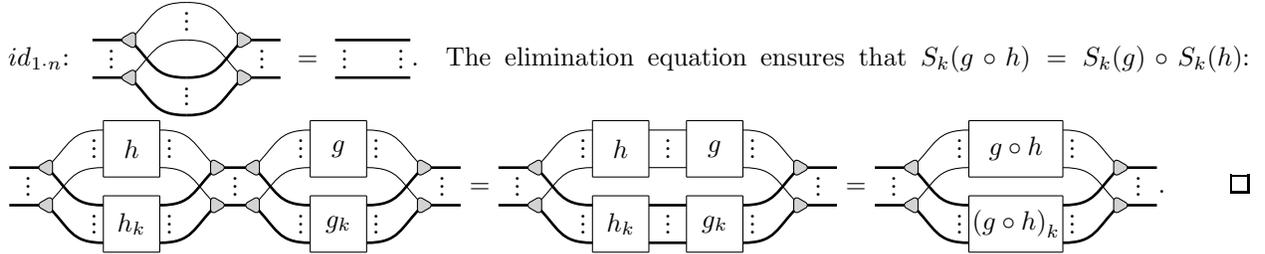
\begin{proof}[Proof of Lemma \ref{wstrip}]\phantomsection\label{pr:wstrip}
	Given a generator $x\in\Sigma_\mathcal{P}$ we have ${S_1(\pi_\mathcal{P}(x))}={\pi_\mathcal{SP}(x)}=\pi_\mathcal{P}$.
\end{proof}

\begin{proof}[Proof of Lemma \ref{struct}]\phantomsection\label{pr:struct}
	
	We use the same method as in the prof of the rewiring theorem. In $\mathcal{S}\textbf{L}$, there are new situations: $s$ represent a simple generator and $S$. The new possible situations for a simple wires are $(i,s)$, $(d,s)$, $(s,s)$, $(s,g)$ and $(s,o)$. For a big wire the new situations are $(i,S)$, $(d,S)$, $(g,S)$, $(S,S)$, $(S,g)$, $(S,d)$ and $(S,o)$.
	
	A diagram of $\mathcal{S}\textbf{L}$ is in expanded form if it contains no big wire in one of the bad situations which are $(g,d)$, $(d,g)$, $(i,S)$, $(d,S)$, $(g,S)$, $(S,S)$, $(S,g)$, $(S,d)$ and $(S,o)$.
	
	So an expanded diagram contains no big generators and is of the form $\omega=\Gamma_b\circ \nu \circ\Delta_a$. Where $\nu$ contains no big wire, so $S_1(\{\nu\})=\nu$. Moreover $S_1(\{\omega\})=S_1(\{ \Gamma_b\circ \nu \circ\Delta_a \})=S_1(\{\Gamma_b \})\circ S_1(\{\nu\}) \circ S_1(\{\Delta_a \})=S_1(\{\nu\})=\nu$.
	So for an expanded diagram $\omega=\Gamma_b\circ S_1(\{\omega\}) \circ\Delta_a$. 
	
	It remains to show that any diagram can be rewritten in expanded form. We proceed by induction on the size of the biggest big wire in a bad situation.
	
	If there is no such wire then we are already in expanded form. Else, we consider the biggest wires in a bad situations.
	
	First we apply the expansion rule to remove all the biggest wires in situation $(i,S)$, $(d,S)$, $(S,S)$, $(S,g)$ and $(S,o)$. If a wire is in situation $(g,d)$ then the elimination rule \ref{elim} can be applied and decreases strictly the size of the big wires in bad situations. If a wire is in situation $(g,d)$ then the expantion rule \ref{exp} can be applied and decreases strictly the size of the big wires in bad situations. If a wire is in situation $(g,S)$ or $(S,d)$ we apply the corresponding unfold equation. This decreases strictly the size of the big wires in bad situations. Thus all the wire in bad situations can be removed and replace by wire with strictly smaller size. Then by induction all diagrams can be rewritten into expanded form. 
	
\end{proof}

\begin{proof}[Proof of Lemma \ref{univ}]\phantomsection\label{pr:univ}
	The diagram clearly commutes. Now let $f: \textbf{K}\to \textbf{1}_{\N} $ and $g: \textbf{K}\to \textbf{L}$ be two functors such that $[\_]\circ f=! \circ g$. We define a functor $h:\textbf{K}\to \mathcal{S}\textbf{L}$. Given a morphism $t\in \textbf{K}[a,b]$ we take $h(t)=\Gamma_{f(a)}\circ S_1 (g(t)) \circ \Delta_{f(b)}$. This is well defined since $[f(a)]=g(a)$. We have $\{h(t)\}=\{\Gamma_{f(a)}\circ S_1 (g(t)) \circ \Delta_{f(b)}\}=\{S_1 (g(t))\}=g(t)$ and $!(h(t))=!(\Gamma_{f(a)}\circ S_1 (g(t)) \circ \Delta_{f(b)})=!_{f(a),f(b)}=f(t)$. Now let $l:\textbf{K}\to \mathcal{S}\textbf{L}$ be another functor such that $\{\_\}\circ l=g$ and $! \circ l=f$. we have $l(t)\in \mathcal{S}\textbf{L}[l(a),l(b)]$ the structure theorem gives us: $l(t)=\Gamma_{l(a)}\circ S_1 (\{l(t)\}) \circ \Delta_{l(b)}=\Gamma_{f(a)}\circ S_1 (g(t)) \circ \Delta_{f(b)}$. So $l=h$. The diagram is a pullback square.
\end{proof}

\begin{proof}[Proof of Lemma \ref{part}]\phantomsection\label{pr:part}
	By construction $\textbf{C}_p$ is the pullback of $!:\textbf{C}\to \mathbf{N}$ and $[\_]:\mathbf{1}_{\N}\to \mathbf{N}$. Thus $\interp{\_}$ lift to a unique functor $\interp{\_}_p:\mathcal{L}\textbf{L}\to \textbf{C}_p$ satisfying $\{\_\} \circ \interp{\_}_p=\interp{\_}\circ \{\_\}$. This functor is an $\N$-colored prop morphism. Furthermore since $!$ and $\{\_\}$ are jointly monic then $\interp{\_}_p$ is faithful iff $\interp{\_}$ is faithful.
\end{proof}

\begin{proof}[Proof of Lemma \ref{equiv}]\phantomsection\label{pr:equiv}
	We consider the wire striper functor $\{\_\}:\mathcal{S}\textbf{L}\to \textbf{L}$. It is clearly essentially surjective. It is also full and faithful since it induces a bijection between $\mathcal{S}\textbf{L}[\mathbf{a},\mathbf{b}]$ and $\textbf{L}[[\mathbf{a}],[\mathbf{b}]]$ by the structure theorem. So it is an equivalence of category and $\textbf{L}\simeq \mathcal{S}\textbf{L}$.
\end{proof}

\subsection{Proofs for section 3}

\begin{proof}[Proof of Lemma \ref{stbox}]\phantomsection\label{pr:stbox}
For a diagram $\omega: \mathbf{a}\to \mathbf{b}$ we define $\omega'\df \Delta_\mathbf{a} \circ \omega \circ \Gamma_\mathbf{b}$.	
We have $\Gamma_\mathbf{a} \circ B(O(\omega)) \circ \Delta_\mathbf{b}= \Gamma_\mathbf{a} \circ B(O(\Gamma_\mathbf{a} \circ \omega' \circ \Delta_\mathbf{b}))\circ \Delta_\mathbf{b}= \Gamma_\mathbf{a} \circ B(O(\omega'))\circ \Delta_\mathbf{b}$. Thus we only need to show that $B(O(\omega'))=\omega'$ for all $\omega': n\cdot 1\to m\cdot 1$.

A diagram of $Box\textbf{P}$ of type $n\cdot 1\to m\cdot 1$ is in \textbf{boxed form} if it is of the form $\Delta_m \circ \Box_f \circ \Gamma_n$.
	
A diagram in boxed form satisfies $B(O(\omega'))=\omega'$ thus we just have to show that any $\omega':n\cdot 1\to m\cdot 1$ can be rewrite into boxed form.
	
First we use the $swap$ rule to transform all swaps into boxes. Then the diagram can be put in a sequence of the form $w_0 \circ b_0 \circ \dots \circ w_{h} \circ b_h \circ w_{h+1}$, where the $w_i$s are in $\textbf{W}$ and the $b_i$s are of the form $id_\mathbf{c} \otimes \Box_f \otimes id_\mathbf{d}$. We use the tensor rule on each $b_i$ until we obtain a new sequence $w'_0 \circ \Box_{f_0} \circ \dots \circ w'_{h} \circ \Box_{f_h} \circ w'_{h+1}$. The rewiring theorem gives us $\omega'=\Delta_m \circ \Box_{f_0} \circ \dots \circ \Box_{f_h} \circ \Gamma_n$.  Finally setting $f=f_0 \circ \dots \circ f_h$ the composition rule gives $\omega'=\Delta_m \circ \Box_f \circ \Gamma_n$ a diagram in boxed form.
\end{proof}

\begin{proof}[Proof of Lemma \ref{same}]\phantomsection\label{pr:same}
	$Unbox$ is an $\N$-colored prop morphism so it is essentially surjective. 
	
	Let $\alpha,\beta \in Box\textbf{L}[\mathbf{a},\mathbf{b}]$ be two diagrams such that $Unbox(\alpha)=Unbox(\beta)$. We have $S_1(O(\alpha))=S_1(O(\beta))$, then applying the wire stripper functor $O(\alpha)=O(\beta)$. The structure theorem of $Box\textbf{L}$ finally gives us $\alpha=\beta$. $Unbox$ is faithful.
	
	Given $\omega\in \mathcal{S}\textbf{L}[\mathbf{a},\mathbf{b}]$, the structure theorem gives us $\omega=\Gamma_\mathbf{b}\circ S_1(\{\omega\}) \circ\Delta_\mathbf{a} $
	
	Then $Unbox(\Gamma_\mathbf{b}\circ B(\{\omega\}) \circ\Delta_\mathbf{a})=\Gamma_\mathbf{b}\circ S_1(O(B(\{\omega\}))) \circ\Delta_\mathbf{a}= \Gamma_\mathbf{b}\circ S_1(\{\omega\}) \circ\Delta_\mathbf{a}=\omega$. $Unbox$ is full.
\end{proof}

%\subsection{Proofs for section 4}

\end{document}